\documentclass[11pt,letterpaper]{article}

\usepackage{color}
\usepackage{marvosym}
\usepackage{appendix}

\usepackage[margin = 1in]{geometry} %

\usepackage{amsmath, amsfonts, amsthm}
\usepackage{latexsym}
\usepackage{comment}
\usepackage{hyperref}
\usepackage{enumerate}
\usepackage{float}
\floatstyle{plaintop}
\restylefloat{table}
\usepackage[]{caption}
\usepackage{graphicx}
\usepackage{epstopdf}
\usepackage{subcaption}
\usepackage[T1]{fontenc}
\usepackage{cite}

\usepackage{wrapfig}

\let\OLDthebibliography\thebibliography
\renewcommand\thebibliography[1]{
	\OLDthebibliography{#1}
	\setlength{\parskip}{0pt}
	\setlength{\itemsep}{0pt plus 0.3ex} 
}

\newtheorem{definition}{Definition}

\newtheorem{lemma}{Lemma}
\newtheorem{theorem}{Theorem}
\newtheorem{proposition}{Proposition}

\graphicspath {{figures/}}

\title{Eternally Dominating Large Grids\footnote{A preliminary version of this paper appeared in the \emph{10th International Conference in Algorithms and Complexity} (CIAC), Springer LNCS, vol. 10236, pp. 393--404, 2017.}}

\author{
	Ioannis Lamprou, 
	Russell Martin,
	Sven Schewe
	\\
	\texttt{\{Forename.Surname\}@liverpool.ac.uk}
	\\ \\
	Department of Computer Science,
	\\
	University of Liverpool, Liverpool, UK 
}

\begin{document}
	
	\maketitle
	
	\begin{abstract} 
		In the m-\emph{Eternal Domination} game, a team of guard tokens initially occupies a dominating set on a graph $G$.
		An attacker then picks a vertex without a guard on it and attacks it. %
		The guards defend against the attack: one of them has to move to the attacked vertex, while each remaining one can choose to move to one of his neighboring vertices. 
		The new guards' placement must again be dominating.
		This attack-defend procedure continues eternally.
		The guards win if they can eternally maintain a dominating set against any sequence of attacks, otherwise the attacker wins.
		
		The m-\emph{eternal domination number} for a graph $G$ is the minimum amount of guards such that they win against any attacker strategy in $G$ (all guards move model).
		We study rectangular grids and provide the first known general upper bound on the m-eternal domination number for these graphs.
		Our novel strategy implements a square rotation principle and eternally dominates $m \times n$ grids by using approximately $\frac{mn}{5}$ guards, which is asymptotically optimal even for ordinary domination.
		\\[10pt]
	\end{abstract} 
	
	\section{Introduction}
	As a natural goal in military deterrence and defence strategies, patrolling a network has always remained topical throughout history.
	In the context of graph searching, such a patrolling task is often modeled as a combinatorial pursuit-evasion game played on a graph.
	This paper studies such a game for a task that requires the eternal domination of a given network.
	
	The \emph{Roman Domination} problem was introduced in \cite{Stewart}: 
	where should Emperor Constantine the Great have located his legions in order to optimally defend against attacks in unsecured locations without leaving another location unsecured?
	In graph theoretic terms, the interest is in producing a \emph{dominating set} of the graph, i.e., a guard placement where each vertex must have a guard on it or on at least one of its neighbors, with possibly some extra problem-specific qualities.
	Some seminal work on this topic includes \cite{Henning, ReVelle}. 
	
	The above model caters only for a single attack on an unsecured vertex.
	A natural question is to consider special domination strategies against a sequence of attacks on the same graph \cite{Burger2}. 
	In this setting, (some of) the guards are allowed to move after each attack to defend against it and modify their overall placement. 
	The difficulty here lies in establishing a guards' placement in order to retain domination after coping with each attack. 
	Such a sequence of attacks can be of finite, i.e., a set of $k$ consecutive attacks, or even \emph{infinite} length.
	
	In this paper, we focus on the latter.
	We wish to protect a graph against attacks happening indefinitely on its vertices.
	Initially, the guards are placed on some vertices of the graph such that they form a dominating set, \emph{with at most one guard per vertex}.
	Then, an attack occurs on an unoccupied vertex. 
	All the guards (may) now move in order to counter the attack:
	One of them moves to the attacked vertex, while each of the others moves to an adjacent vertex of theirs such that the new guards' placement again forms a dominating set.
	This takes place ad infinitum.
	
	The attacker's objective is to devise a sequence of attacks, which leads the guards to a non-dominating placement.
	On the other hand, the guards wish to maintain a sequence of dominating sets without any interruption.
	The m-\emph{Eternal Domination} problem, studied in this paper, deals with determining the minimum number of guards such that they eternally protect the graph in the above fashion.
	The focus is on rectangular grids, where, to the best of our knowledge, we provide a first general upper bound.
	
	\paragraph{Related Work.} 
	Infinite order domination was first considered by Burger et al. \cite{Burger} as an extension to finite order domination.
	Later on, Goddard et al. \cite{Goddard} proved some first bounds with respect to other graph-theoretic notions 
	(like independence and clique cover) for the one-guard-moves and all-guards-move cases.
	The relationship between eternal domination and clique cover is examined more carefully in \cite{Anderson}.
	There exists a series of other papers with several combinatorial bounds, e.g., see \cite{Goldwasser, Henning2, Henning3, Klos1}. 
	
	Regarding the special case of grid graphs, Chang \cite{Chang} gave many strong upper and lower bounds for the domination number.
	Indeed, Gon\c{c}alves et al. \cite{Goncalves} proved Chang's construction optimal for rectangular grids where both dimensions are greater or equal to $16$.
	Moving onward to eternal domination, bounds for $3\times n$ \cite{3n, Messinger}, $4\times n$ \cite{4n} and $5\times n$ \cite{5n} grids have been examined,
	where the bounds are almost tight for $3 \times n$ and exactly tight for $4\times n$.
	
	Due to the mobility of the guards in eternal domination and the breakdown into alternate turns (guards vs attacker),
	one can view this problem as a pursuit-evasion combinatorial game in the same context as \emph{Cops \& Robber} \cite{CopRobber} and the \emph{Surveillance Game} \cite{Surv, Surv2}.
	In all three of them, there are two players who alternately take turns, with one of them pursuing the other possibly indefinitely.
	
	An analogous \emph{Eternal Vertex Cover} problem has been considered \cite{Fomin, Klos2, Klos3}, where attacks occur on the edges of the graph.
	In that setting, the guards defend against an attack by traversing the attacked edge, while they move in order to preserve a vertex cover after each turn.
	
	Recently, the m-\emph{eviction number} is studied in \cite{NewProblem}, where attacks occur on the vertices occupied by guards and they have to move to survive, whilst always maintaining a dominating set.
	
	For an overall picture and further references on the topic, we refer the reader to a recent survey on graph protection \cite{Survey}.
	
	\paragraph{Our Result.}
	We make a first step towards answering an open question raised by Klostermeyer and Mynhardt \cite{Survey}: We show that, in order to ensure m-eternal domination in rectangular grids,
	only a linear number of extra guards is needed compared to domination. 
	
	To obtain this result, we devise an unravelling strategy of successive (counter) clockwise rotations for the guards to eternally dominate an infinite grid. 
	This strategy is referred to as the \emph{Rotate-Square} strategy.
	Then, we apply the same strategy to finite grids with some extra guards to ensure the boundary remains always guarded. 
	Overall, we show that $\lceil \frac{mn}{5} \rceil + \mathcal{O}(m+n)$ guards suffice to eternally dominate an $m \times n$ grid, for $m, n \ge 16$.
	
	\paragraph{Outline.}
	
	In Section~\ref{sec:pre}, we define some basic graph-theoretic notions and m-\emph{Eternal Domination} as a two-player combinatorial pursuit-evasion game.
	Later, in Section~\ref{sec:inf}, we describe the basic components of the Rotate-Square strategy and prove that it can be used to dominate an infinite grid forever. 
	Later, in Section~\ref{sec:fin} we show how the strategy can be adjusted to eternally dominate finite grids by efficiently handling movements near the boundary and the corners.
	Finally, in Section~\ref{sec:con}, we conclude with some final remarks and open questions.
	
	\section{Preliminaries}\label{sec:pre}
	
	Let $G = (V(G),E(G))$ be a simple connected undirected graph.
	We denote an edge between two connected vertices, namely $v$ and $u$, as $(u, v) \in E(G)$ (or equivalently $(v,u)$).
	The \emph{open-neighborhood} of a subset of vertices $S \subseteq V(G)$ is defined as $N(S) = \{v \in V(G)\setminus S: \exists u \in S \text{ such that } (u,v) \in E(G)\}$ and
	the \emph{closed-neighborhood} as $N[S] = S \cup N(S)$.
	For a single vertex $v \in V(G)$, we simplify the notation $N(\{v\})$ to $N(v)$ and, similarly, $N[\{v\}]$ to $N[v]$. A \emph{path} of length $n-1 \in \mathbb{N}$, namely $P_n$, is a graph where $V(P_n) = \{v_0, v_1, \ldots, v_{n-1}\}$ and
	$E(P_n) = \{(v_0, v_1), (v_1, v_2), \ldots (v_{n-2}, v_{n-1})\}$.
	The \emph{Cartesian product} of two graphs $G$ and $H$ is another graph denoted $G \Box H$ where
	$V(G\Box H) = V(G) \times V(H)$ and two vertices  $(v,v')$  and $(u,u')$ are adjacent if either $v = u$ and $(v',u') \in E(H)$ or $v' = u'$ and $(v, u) \in E(G)$.
	A \emph{grid}, namely $P_m \Box P_n$, is the Cartesian product of two paths of lengths $m-1, n-1 \in \mathbb{N}$. 
	
	A set of vertices $S \subseteq V(G)$ is called a \emph{dominating set} of $G$ if $N[S] = V(G)$.
	That is, for each $v \in V(G)$, either $v \in S$ or there exists a vertex $u \in S$ $(u \neq v)$ such that $(u, v) \in E(G)$.
	A minimum-size such set, say $S^*$, is called a \emph{minimum dominating set} of $G$ and $\gamma(G) = |S^*|$  is defined as
	the \emph{domination number} of $G$.
	For grids, we simplify the notation $\gamma(P_m \Box P_n)$ to $\gamma_{m,n}$.
	
	\emph{Eternal Domination} can be regarded as a combinatorial pursuit-evasion game played on a graph $G$. 
	There exist two players: one of them controls the \emph{guards}, while the other controls the \emph{attacker}. 
	The game takes place in \emph{rounds}. Each round consists of two \emph{turns}: one for the guards and one for the attacker.
	
	Initially (round $0$), the guard tokens are placed such that they form a dominating set on $G$. 
	Then, without loss of generality, the attacker attacks a vertex without a guard on it. 
	A guard, dominating the attacked vertex, must now move on it to counter the attack.
	Notice that at least one such guard exists because their initial placement is dominating.
	Moreover, the rest of the guards may move; a guard on vertex $v$ can move to any vertex in $N[v]$.
	The guards wish to ensure that their modified placement is still a dominating set for $G$.
	The game proceeds in a similar fashion in any subsequent round. 
	Guards win if they can counter any attack of the attacker and eternally maintain a dominating set; that is, for an infinite number of attacks.
	Otherwise, the attacker wins, as he manages to force the guards to reach a placement that is no longer dominating;
	then, an attack on an undominated vertex suffices to win.
	From now on, we say that a vertex is \emph{unoccupied} when no guard lies on it.
	
	\begin{definition}
		$\gamma^\infty_{\mathrm{m}}(G)$ stands for the \emph{m-eternal domination number} of a graph $G$, 
		i.e., the minimum size of a guards' team that can eternally dominate $G$ (when \emph{all} guards can move at each turn).
	\end{definition}
	
	As above, we simplify $\gamma^\infty_{\mathrm{m}}(P_m \Box P_n)$ to $\gamma^\infty_{m,n}$.
	Since the initial guards' placement is dominating, we get $\gamma^\infty_{\mathrm{m}}(G) \ge \gamma(G)$ for any graph $G$.
	By a simple rotation, we get $\gamma_{m,n} = \gamma_{n, m}$ and $\gamma^\infty_{m,n} = \gamma^\infty_{n,m}$.
	Finally, multiple guards are \emph{not} allowed to lie on a single vertex, since this could provide an advantage for the guards \cite{note}.
	
	\section{Eternally Dominating an Infinite Grid} \label{sec:inf}
	In this section, we describe a strategy to eternally dominate an \emph{infinite grid}. 
	We denote an infinite grid as $G_\infty$ and define it as a pair $(V(G_\infty), E(G_\infty))$, 
	where $V(G_\infty) = \{(x, y): x, y \in \mathbb{Z}\}$ and any vertex $(x, y) \in V(G_\infty)$ is adjacent to $(x,y-1)$, $(x, y+1)$, $(x-1, y)$ and $(x+1, y)$.
	In all figures to follow, we view the grid as a mesh, i.e.,\ similar to a chessboard, where each \emph{cell} of the mesh corresponds to a vertex of $V(G_\infty)$ and neighbors four other cells: the one above, below, left and right of it.
	We assume row $x$ is \emph{above} row $x+1$ and column $y$ is \emph{left} of column $y+1$.
	Each guard occupies a single cell and has the capability to move to an adjacent cell (left, right, up or down) during the guards' turn.
	For a visual explanation of the grid-mesh correspondence, see Figure~\ref{fig:corr}.
	
	\begin{figure}[h]
		\centering
		\begin{tabular}{cc}
			\begin{subfigure}{0.45\linewidth}
				\centering
				\includegraphics[scale= 2.25]{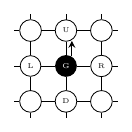}
			\end{subfigure}
			&
			\begin{subfigure}{0.45\linewidth}
				\centering
				\includegraphics[scale= 2.25]{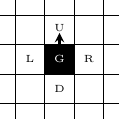}
			\end{subfigure}
			\\
			\begin{minipage}{.45\textwidth}
				\leavevmode\subcaption{Some local view of a grid graph consisting of vertices and edges}
			\end{minipage}
			&
			\begin{minipage}{.45\textwidth}
				\leavevmode\subcaption{The local view of the corresponding mesh configuration consisting of cells}
			\end{minipage}	
		\end{tabular}
		\caption{From a grid graph (a) to a mesh configuration (b): a guard lying on vertex/cell $G$ can move to any of its neighboring vertices/cells $L$, $R$, $U$, $D$ during the guards' turn}
		\label{fig:corr}
	\end{figure}
	
	Initially, let us consider a family of dominating sets for $G_\infty$.
	In the following, let $\mathbb{Z}^2 := \mathbb{Z} \times \mathbb{Z}$ and let $\mathbb{Z}_5 := \{0, 1, 2, 3, 4\}$ stand for the group of integers modulo $5$.
	We then define the function $f:\mathbb{Z}^2 \rightarrow \mathbb{Z}_5$ as $f(x, y) = x + 2y \;(\bmod\; 5)$ for any $(x, y) \in \mathbb{Z}^2$.
	This function appears in \cite{Chang} and is central to providing an optimal dominating set for sufficiently large finite grids.
	Now, let $D_{t} = \{(x,y) \in V(G_\infty) : f(x,y) = t\}$ for $t \in \mathbb{Z}_5$ and $\mathcal{D}(G_\infty) = \{D_{t}: t \in \mathbb{Z}_5\}$.
	For purposes of symmetry, let us define $f'(x, y) = f(y,x)$ and then let $D'_t =  \{(x,y) \in V(G_\infty): f'(x,y) = t\}$ and $\mathcal{D}'(G_\infty) = \{D'_{t}: t \in \mathbb{Z}_5\}$.
	
	
	\begin{proposition}\label{prop:D_t}
		Any $D_t, D'_t \in \mathcal{D}(G_\infty) \cup \mathcal{D}'(G_\infty)$ is a dominating set for $G_\infty$.
	\end{proposition}
	\begin{proof}
		Let $(x, y) \in V(G_\infty)$ and $f(x, y) = t \in \{0, 1, 2, 3, 4\}$. We consider all possible cases for another vertex $(w,z) \in V(G_\infty)$:
		\begin{itemize}
			\item If $f(w,z) = t$, then $(w,z) \in D_{t}$.
			\item If $f(w,z) = t + 1 \;(\bmod\; 5)$, then $f(w-1, z) = t$ and so $(w-1, z) \in D_{t}$ dominates $(w,z)$.
			\item If $f(w,z) = t - 1 \;(\bmod\; 5)$, then $f(w+1,z) = t$ and so $(w+1, z) \in D_{t}$ dominates $(w,z)$.
			\item If $f(w,z) = t + 2 \;(\bmod\; 5)$, then $f(w,z-1) = t$ and so $(w,z-1)\in D_{t}$ dominates $(w,z)$.
			\item If $f(w,z) = t - 2 \;(\bmod\; 5)$, then $f(w,z+1) = t$ and so $(w,z+1)\in D_{t}$ dominates $(w,z)$.
		\end{itemize}
		Similarly, let $(x, y) \in V(G_\infty)$ and $f'(x, y) = t \in \{0, 1, 2, 3, 4\}$. Again, we consider all possible cases for another vertex $(w,z) \in V(G_\infty)$:
		\begin{itemize}
			\item If $f'(w,z) = t$, then $(w,z) \in D'_{t}$.
			\item If $f'(w,z) = t + 1 \;(\bmod\; 5)$, then $f'(w, z-1) = t$ and so $(w, z-1) \in D'_{t}$ dominates $(w,z)$.
			\item If $f'(w,z) = t - 1 \;(\bmod\; 5)$, then $f'(w, z+1) = t$ and so $(w, z+1) \in D'_{t}$ dominates $(w,z)$.
			\item If $f'(w,z) = t + 2 \;(\bmod\; 5)$, then $f'(w-1, z) = t$ and so $(w-1,z)\in D'_{t}$ dominates $(w,z)$.
			\item If $f'(w,z) = t - 2 \;(\bmod\; 5)$, then $f'(w+1, z) = t$ and so $(w+1,z)\in D'_{t}$ dominates $(w,z)$.
		\end{itemize}
	\end{proof}
	Notice that the above constructions form \emph{perfect} dominating sets, i.e., dominating sets where each vertex is dominated by \emph{exactly} one vertex (possibly itself),
	since, for each vertex $v \in V(G_\infty)$, exactly one vertex from $N[v]$ lies in $D_t$ (respectively $D'_t$) by the definition of $D_t$ (respectively $D'_t$). 
	
	\subsection{A First Eternal Domination Strategy}
	Let us now consider a \emph{shifting-style} strategy as a first simple strategy to eternally dominate $G_\infty$.
	The guards initially pick a placement $D_t$ for some $t \in \mathbb{Z}_5$.
	Next, an attack occurs on some unoccupied vertex.
	Since the $D_t$ placement perfectly dominates $G_\infty$, there exists exactly one guard adjacent to the attacked vertex.
	Therefore, it is mandatory for him to move onto the attacked vertex.
	His move defines a direction in the grid: left, right, up or down.
	The rest of the strategy reduces to each guard moving according to the defined direction.

	The aforementioned strategy works fine for the infinite grid.
	Nonetheless, applying it (directly or modified) to a finite grid encounters many obstacles.
	Shifting the guards toward one course leaves some vertices in the very end of the opposite course (near the boundary) undominated, since there is no longer an unlimited supply of guards to ensure protection.
	To overcome this problem, we propose a different strategy whose main aim is to redistribute the guards without creating any bias to a specific direction.

	\subsection{Unoccupied Squares}
	Another m-eternal domination strategy is to \emph{rotate} the guards' placement around subgrids of size $2\times2$, in which all four cells are unoccupied. 
	We refer to such a subgrid as an \emph{unoccupied square}.
	Intuitively, by using such an approach, the overall movement is zero and
	the guards always occupy a placement in $\mathcal{D}(G_\infty) \cup \mathcal{D}'(G_\infty)$ after an attack is defended.
	
	Consider some vertex $(x, y) \in V(G_\infty)$, where $(x, y) \in D_t$ for some value $t \in \mathbb{Z}_5$.
	Now, assume that the guards lie on the vertices specified in $D_t$ and hence form a dominating set.
	In Figure~\ref{fig:sq}, we partially view $G_\infty$ where the black cell represents a guard on some cell $(x,y) \in D_t$ and the gray cells represent guards elsewhere in $D_t$.
	By looking around $(x,y)$, we identify the existence of four \emph{unoccupied squares}.
	For $i = 0, 1, 2, 3$, let $SQ_i(x,y)$ denote the $i$-th unoccupied square with respect to $(x,y)$.
	
	{\small
		\begin{itemize}
			\item $SQ_0(x,y) = \{(x-1, y+1), (x-1, y+2), (x, y+1), (x, y+2)\}$
			\item $SQ_1(x,y) = \{(x+1, y), (x+1, y+1), (x+2, y), (x+2, y+1)\}$
			\item $SQ_2(x,y) = \{(x, y -2), (x, y-1), (x+1,y-2), (x+1, y-1) \}$
			\item $SQ_3(x,y) = \{(x-2, y-1), (x-2, y), (x-1, y-1),(x - 1, y)\}$
		\end{itemize}
	}
	In Figure~\ref{fig:sq}, a cell in an unoccupied square $SQ_i(x,y)$ has a label $SQ_i$.

	One can verify that, for every $(w,z) \in \bigcup_{i = 0}^3 SQ_i(x,y)$, we get $f(w,z) \neq f(x,y)$ and thus $(w,z) \notin D_t$.
	Notice that $(x, y)$ has exactly one adjacent cell in each of these unoccupied squares and is the only guard that dominates these four cells, since domination is perfect.
	We say that a guard on $(x,y)$ \emph{slides along} the side of an unoccupied square $SQ_i(x,y)$ when he moves to cell $(w,z)$ adjacent to $(x,y)$, where $(w,z)$ is also adjacent to a cell in $SQ_i(x,y)$.
	In other words, the $(x,y)$-guard's current and previous cells are both adjacent to a cell in $SQ_i(x,y)$.
	In the case of a $D_t$ placement, see Figure~\ref{fig:sq}, an attack on a cell $(w,z) \in SQ_i(x,y) \cap N((x,y))$ would mean the guard on $(x,y)$ moves to $(w,z)$ and slides along the side of unoccupied square $SQ_{(i+1) \bmod 4}(x,y)$.
	For example, an attack on the bottom-right cell of $SQ_3$ would mean the $(x,y)$-guard slides along $SQ_0$; see Figures~\ref{fig:slides1}, \ref{fig:slides2}.
	
	\begin{figure}
		\centering
		\begin{minipage}{\linewidth}
			\begin{subfigure}{\textwidth}
				\centering
				\includegraphics[scale = 0.75]{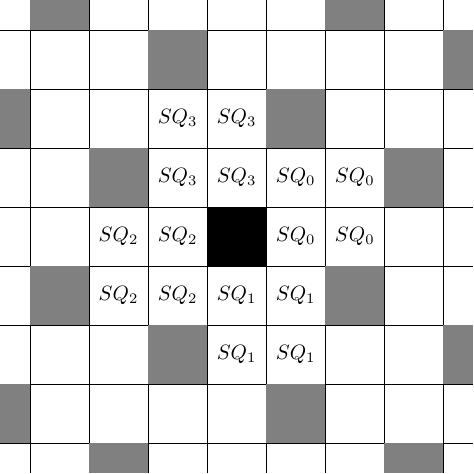}
				\caption{Unoccupied squares around $(x, y) \in D_t$; $(x,y)$ in black} 
				\label{fig:sq}
			\end{subfigure}
		\end{minipage} \\[1.25em]
		\begin{minipage}{\linewidth}
			\begin{subfigure}{\textwidth}
				\centering
				\includegraphics[scale = 0.75]{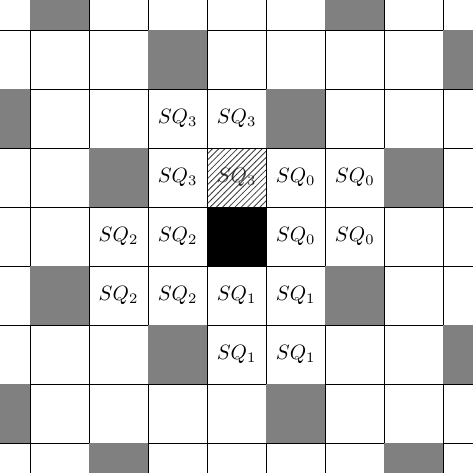}
				\caption{Attack on bottom-right cell of $SQ_3(x,y)$: we identify the defence-responsible guard $(x,y)$ (in black) and the corresponding unoccupied squares} 
				\label{fig:slides1}
			\end{subfigure}
		\end{minipage} \\[1.25em]
		\begin{minipage}{\linewidth}
			\begin{subfigure}{\textwidth}
				\centering
				\includegraphics[scale = 0.75]{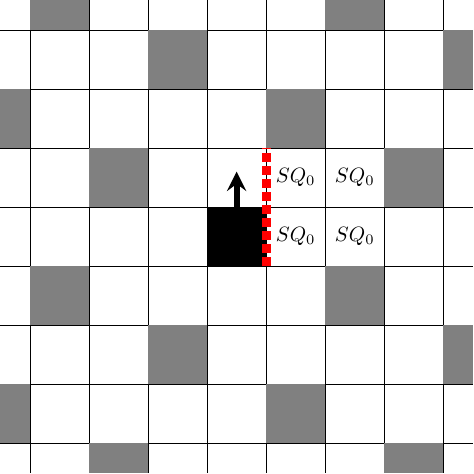}       
				\caption{The defence-responsible guard slides along a side of $SQ_0(x,y)$ (dotted line): its current and next cell are both adjacent to $SQ_0(x,y)$ cells}
				\label{fig:slides2}
			\end{subfigure}
		\end{minipage}
		\caption{Examples of defence-responsible guard, unoccupied squares and sliding along for the $D_t$ case}
		\label{fig:Dt-def}
	\end{figure}
	
	The aforementioned observations also extend to some vertex $(x,y)$ in a dominating set $D'_t$.
	We now define the four unoccupied squares as follows (see Figure~\ref{fig:sq1}):
	{\small
		\begin{itemize}
			\item $SQ'_0(x,y) = \{(x, y+1), (x, y+2), (x+1, y+1), (x+1, y+2)\}$
			\item $SQ'_1(x,y) = \{(x+1, y-1), (x+1, y), (x+2, y-1), (x+2, y)\}$
			\item $SQ'_2(x,y) = \{(x-1, y -2), (x-1, y-1), (x,y-2), (x, y-1) \}$
			\item $SQ'_3(x,y) = \{(x-2, y), (x-2, y+1), (x-1, y),(x - 1, y+1)\}$  
		\end{itemize}
	}
	Similarly to before, the squares are unoccupied, since for every $(w,z) \in \bigcup_{i = 0}^3 SQ'_i$ we get $f'(w,z) \neq f'(x,y)$ and thus $(w,z) \notin D'_t$.
	The $(x,y)$-guard  has exactly one adjacent cell in each of these unoccupied squares and protecting an attack on $SQ'_i$ now means sliding along the side of $SQ'_{(i-1) \bmod 4}$.
	For example, an attack on the bottom-right cell of $SQ'_2$ means the $(x,y)$-guard slides along $SQ'_1$ (Figures~\ref{fig:slides1'}, \ref{fig:slides2'}).
	
	\begin{figure}
		\centering
		\begin{minipage}{\linewidth}
			\begin{subfigure}{\textwidth}
				\centering
				\includegraphics[scale = 0.75]{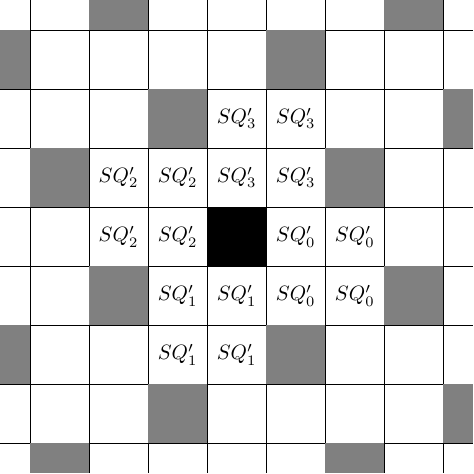}
				\caption{Unoccupied squares around $(x, y) \in D'_t$; $(x,y)$ in black} 
				\label{fig:sq1}
			\end{subfigure}
		\end{minipage} \\[1.25em]
		\begin{minipage}{\linewidth}
			\begin{subfigure}{\textwidth}
				\centering
				\includegraphics[scale = 0.75]{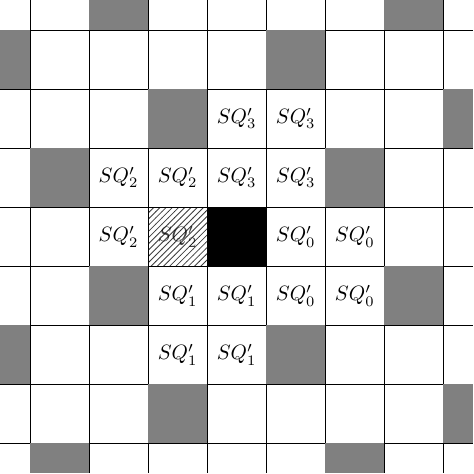}
				\caption{Attack on bottom-right cell of $SQ'_2(x,y)$:  we identify the defence-responsible guard $(x,y)$ (in black) and the corresponding unoccupied squares} 
				\label{fig:slides1'}
			\end{subfigure}
		\end{minipage} \\[1.25em]
		\begin{minipage}{\linewidth}
			\begin{subfigure}{\textwidth}
				\centering
				\includegraphics[scale = 0.75]{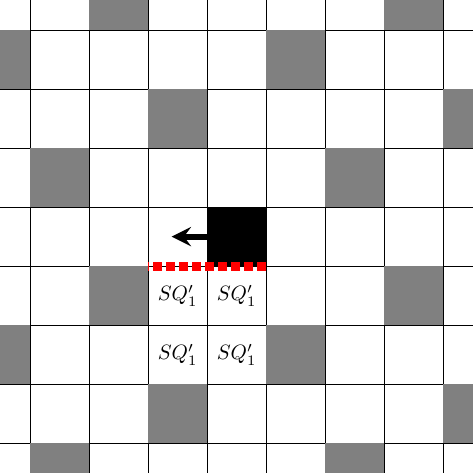}       
				\caption{The defence-responsible guard slides along a side of $SQ'_1(x,y)$ (dotted line): its current and next cell are both adjacent to $SQ'_1(x,y)$ cells}
				\label{fig:slides2'}
			\end{subfigure}
		\end{minipage}
		\caption{Examples of defence-responsible guard, unoccupied squares and sliding along for the $D'_t$ case}
		\label{fig:Dt'-def}
	\end{figure}
	
	\subsection{The Rotate-Square Strategy}
	
	We hereby describe the \emph{Rotate-Square} strategy and prove it eternally dominates $G_\infty$.
	The strategy makes use of the unoccupied squares idea.
	Once an attack occurs on some unoccupied vertex, since any $D_t$ or $D'_t$ dominating set the guards form is perfect, there exists a single guard who is able to defend against it.
	We refer to this guard as the \emph{defence-responsible} guard.
	Without loss of generality, let the defence-responsible guard lie on some cell $(x,y) \in D_t$ for some $t$.
	For $D'_t$ placements, the arguments are similar.
	We identify the four unoccupied squares around the defence-responsible guard as in Figure~\ref{fig:sq}.
	Assume the attack happened on a vertex $(w,z) \in SQ_i(x,y)$.
	Then, as discussed earlier, the defence-responsible guard moves to $(w,z)$ to defend against the attack and such a move means he is sliding along the side of square $SQ_{(i+1) \bmod 4}(x,y)$.
	We refer to $SQ_{(i+1) \bmod 4}(x,y)$ as the \emph{pattern square}, i.e., the unoccupied square along whose side the defence-responsible guard slides to defend the attack.
	
	Notice that, due to the grid topology and the fact that a $D_t$ placement is a perfect dominating set,
	there are exactly four guards adjacent to cells of the pattern square: exactly one guard per cell of the pattern square (one of them being the defence-responsible guard). 
	We refer to these four guards as the \emph{pattern guards}.
	In Figures~\ref{fig:pattern-guards} and \ref{fig:pattern-guards'}, we identify the pattern guards for each potential pattern square out of the four unoccupied squares related to a $D_t$ or $D'_t$ placement.
	Besides the defence-responsible guard, the other three guards dominating the pattern square also slide along a side of the pattern square, such that the four guards' overall movement looks as a rotation step around the pattern square.
	
	\begin{figure}
		\centering
		\begin{minipage}{\linewidth}
			\begin{subfigure}{0.45\textwidth}
				\centering
				\includegraphics[width = \textwidth]{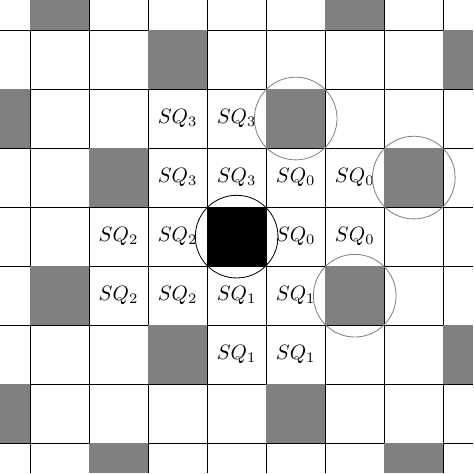}
				\caption{Pattern guards for $SQ_0$} 
				\label{fig:patternSQ0}
			\end{subfigure}
			\hfill
			\begin{subfigure}{0.45\textwidth}
				\centering
				\includegraphics[width = \textwidth]{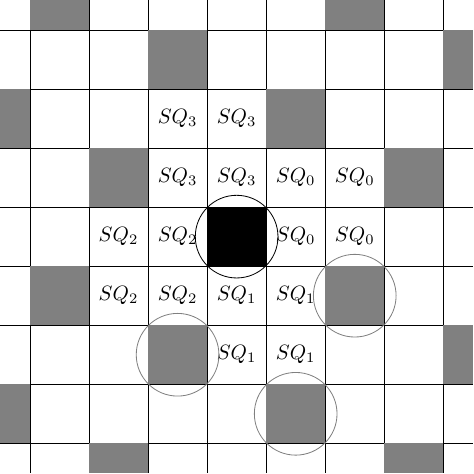}
				\caption{Pattern guards for $SQ_1$} 
				\label{fig:patternSQ1}
			\end{subfigure}
		\end{minipage}\\[3em]
		\begin{minipage}{\linewidth}
			\begin{subfigure}{0.45\textwidth}
				\centering
				\includegraphics[width = \textwidth]{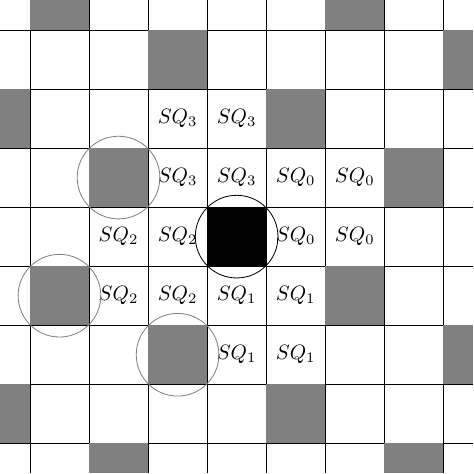}
				\caption{Pattern guards for $SQ_2$} 
				\label{fig:patternSQ2}
			\end{subfigure}
			\hfill
			\begin{subfigure}{0.45\textwidth}
				\centering
				\includegraphics[width = \textwidth]{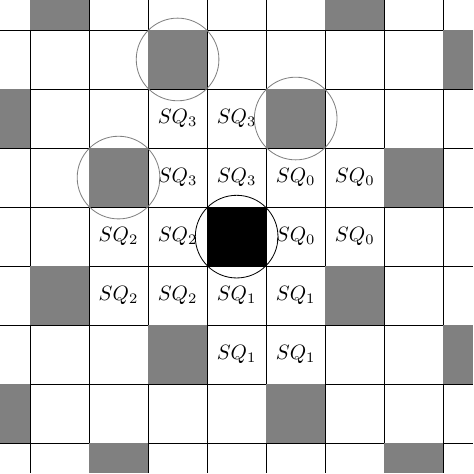}
				\caption{Pattern guards for $SQ_3$} 
				\label{fig:patternSQ3}
			\end{subfigure}
		\end{minipage}\\[1em]	
		\caption{Pattern guards for $D_t$ unoccupied squares (circled)}
		\label{fig:pattern-guards}
	\end{figure}
	
	\begin{figure}
		\centering
		\begin{minipage}{\linewidth}
			\begin{subfigure}{0.45\textwidth}
				\centering
				\includegraphics[width = \textwidth]{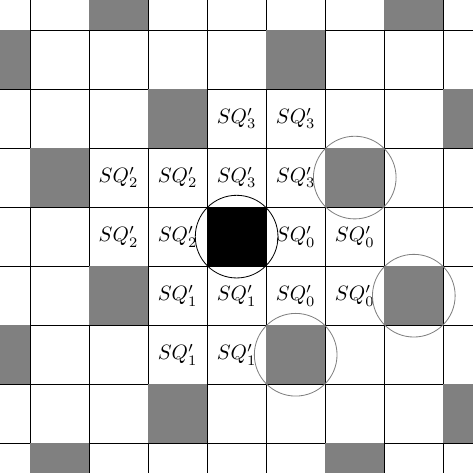}
				\caption{Pattern guards for $SQ'_0$} 
				\label{fig:patternSQ0'}
			\end{subfigure}
			\hfill
			\begin{subfigure}{0.45\textwidth}
				\centering
				\includegraphics[width = \textwidth]{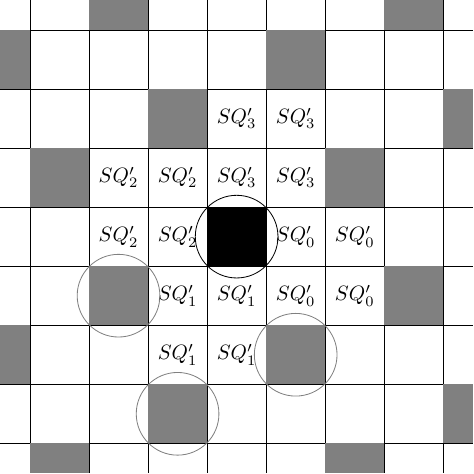}
				\caption{Pattern guards for $SQ'_1$} 
				\label{fig:patternSQ1'}
			\end{subfigure}
		\end{minipage}\\[3em]
		\begin{minipage}{\linewidth}
			\begin{subfigure}{0.45\textwidth}
				\centering
				\includegraphics[width = \textwidth]{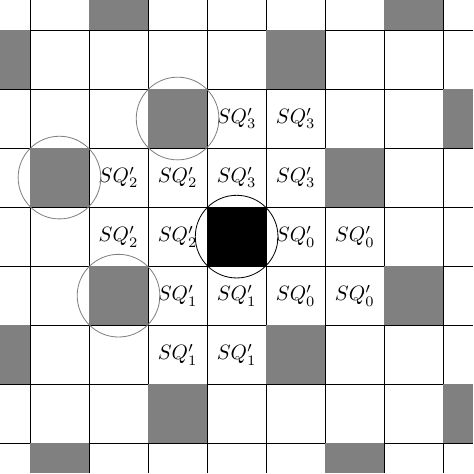}
				\caption{Pattern guards for $SQ'_2$} 
				\label{fig:patternSQ2'}
			\end{subfigure}
			\hfill
			\begin{subfigure}{0.45\textwidth}
				\centering
				\includegraphics[width = \textwidth]{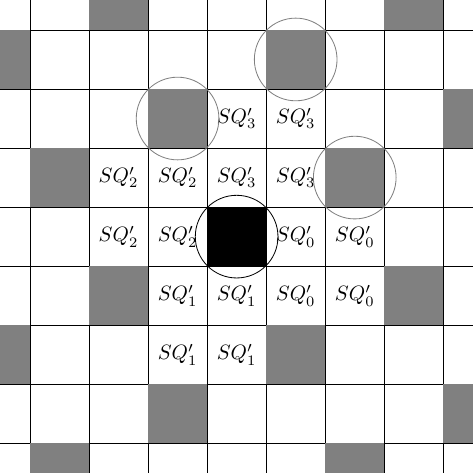}
				\caption{Pattern guards for $SQ'_3$} 
				\label{fig:patternSQ3'}
			\end{subfigure}
		\end{minipage}\\[1em]	
		\caption{Pattern guards for $D'_t$ unoccupied squares (circled)}
		\label{fig:pattern-guards'}
	\end{figure}

	To facilitate a more formal explanation, let us break down the guards' turn into a few distinct components.
	Of course, the guards are always assumed to move concurrently during their turn. 
	
	Initially, the guards are assumed to occupy a dominating set $D_t$ in $\mathcal{D}(G_\infty)$.
	Then, an attack occurs on a vertex in $V(G_\infty) \setminus D_t$.
	To defend against it, the guards apply \emph{Rotate-Square}:
	
	\begin{itemize}
		\item[\textbf{(1)}] Identify the defence-responsible guard.
		\item[\textbf{(2)}] Identify the pattern square and the pattern guards.
		\item[\textbf{(3)}] The pattern guards slide along the sides of the pattern square.
		\item[\textbf{(4)}] Repeat the rotation pattern in horizontal and vertical lanes in hops of distance five.
	\end{itemize}
	
	Let us examine each of these strategy components more carefully.
	
	Step (1) requires looking at the grid and spotting the guard, which lies on a cell adjacent to the attacked cell.
	The four unoccupied squares around the defence-responsible guard are identified as in Figure~\ref{fig:sq}.
	
	In step (2), the pattern square is identified as the unoccupied square along whose side the defence-responsible guard has to slide to defend the attack. 
	The four guards adjacent to cells of the pattern square are identified as the pattern guards. 
	
	In step (3), each of the pattern guards (including the defence-responsible guard) slides along a side of the pattern square.
	For an example, see Figure~\ref{fig:step3}: the defence-responsible guard in cell $(x,y)$ (in black) defends against an attack on the bottom-right cell of $SQ_3(x,y)$ by sliding along $SQ_0(x,y)$.
	Then, the other three guards around $SQ_0(x,y)$ (in gray) slide along a side of $SQ_0(x,y)$ as well.
	The latter happens in order to preserve that the pattern square $SQ_0(x,y)$ remains unoccupied.
	
	\begin{figure}
		\centering
		\begin{minipage}{\linewidth}
			\begin{subfigure}{\textwidth}
				\centering
				\includegraphics[scale = 0.75]{sq.pdf}
				\caption{Unoccupied squares around $(x, y) \in D_t$; $(x,y)$ in black} 
				\label{fig:step3a}
			\end{subfigure}
		\end{minipage} \\[1.5em]
		\begin{minipage}{\linewidth}
			\begin{subfigure}{\textwidth}
				\centering
				\includegraphics[scale = 0.75]{attack-sq3.pdf}
				\caption{Attack on bottom-right cell of $SQ_3(x,y)$: the defence-responsible guard in cell $(x,y)$ (in black) has to slide along a side of $SQ_0(x,y)$, thus $SQ_0(x,y)$ is identified as the pattern square for the next guards' turn} 
				\label{fig:step3b}
			\end{subfigure}
		\end{minipage} \\[1.5em]
		\begin{minipage}{\linewidth}
			\begin{subfigure}{\textwidth}
				\centering
				\includegraphics[scale = 0.75]{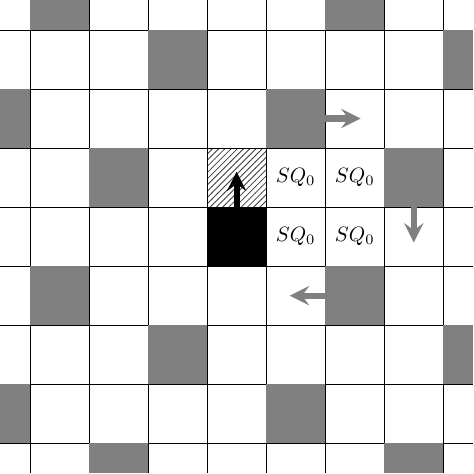}       
				\caption{Pattern guards slide along sides of pattern square $SQ_0(x,y)$; arrows indicate the corresponding directions}
				\label{fig:step3c}
			\end{subfigure}
		\end{minipage}
		\caption{An example for Step (3) of Rotate-Square}
		\label{fig:step3}
	\end{figure}
	
	Eventually, in step (4), the pattern square is used as a guide for the move of the rest of the guards in $D_t$. 
	Consider a pattern guard initially lying on vertex $(w,z)$.
	By construction of $D_t$, guards lie on all vertices $(w \pm 5\alpha, z \pm 5\beta)$ for $\alpha, \beta \in \mathbb{N}$, since adding multiples of $5$ in both dimensions does not affect membership in $D_t$ by definition of $f(\cdot)$. We refer to the set $\{(w \pm 5\alpha, z \pm 5\beta): \alpha, \beta \in \mathbb{N} \}$ as the \emph{cousins} of $(w,z)$.
	Each pattern guard, in step (3), slides along a side of the pattern square. 
	His move defines a direction on the grid: up, down, left, right.
	For each pattern guard $(w,z)$, the strategy of his cousins reduces to taking a step in the same direction.
	The rest of the guards, i.e., guards who are not cousins to any pattern guard, do not move during this turn, i.e., will remain in their pre-attack location after the attack.
	From now on, we refer to these guards as \emph{stationary guards}.
	We provide an example execution of step (4) in Figure~\ref{fig:step4}. 
	The original pattern guards are given in black.
	The circles enclose the repetitions of the pattern guards' move by their cousins, in gray. 
	Guards outside a circle do not move during this turn, i.e., they remain in their pre-attack location after the attack.
	
	\begin{figure}
		\centering
		\includegraphics[scale = 0.75]{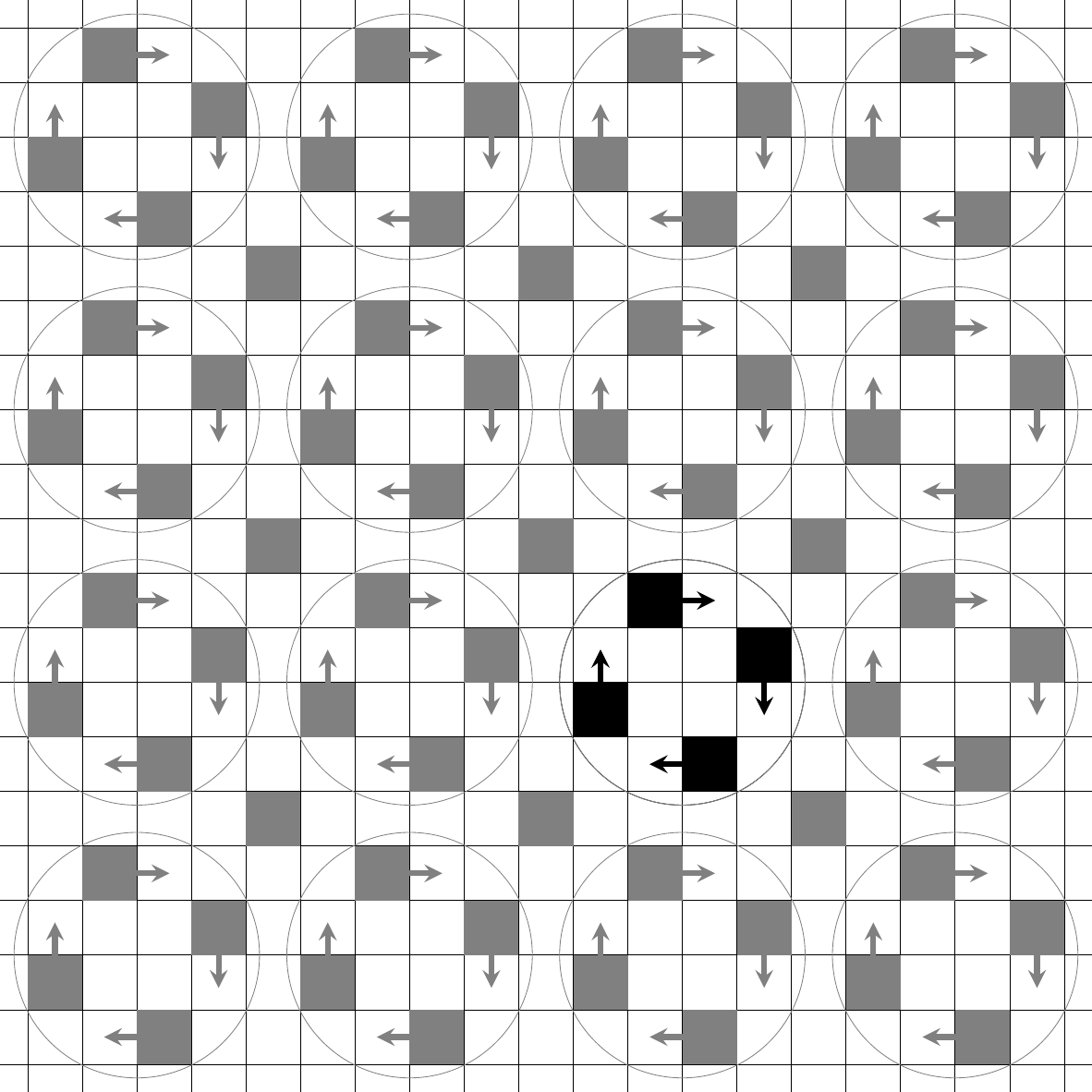}       
		\caption{An example execution for Step (4) of Rotate-Square: pattern guards are given in black}
		\label{fig:step4}
	\end{figure}
	
	\begin{lemma}\label{lem:step}
		Assume the guards occupy a dominating placement $D \subseteq V(G_\infty)$ in $\mathcal{D}(G_\infty) \cup \mathcal{D}'(G_\infty)$ and
		an attack occurs on a vertex in $V(G_\infty) \setminus D$.
		After applying the Rotate-Square strategy, the guards successfully defend against the attack and again form a dominating set in $\mathcal{D}(G_\infty) \cup \mathcal{D}'(G_\infty)$.
	\end{lemma}
	
	\begin{proof}
		In this proof, we are going to demonstrate that any of the four possible attacks (one per unoccupied square) around a vertex in a $D_t$ (or $D'_t$) placement can be defended by Rotate-Square and, most importantly,
		the guards still occupy a placement in $\mathcal{D}(G_\infty) \cup \mathcal{D}'(G_\infty)$ after their turn.
		Below, in Figure~\ref{fig:DtSQ3}, we provide pictorial details for one out of eight cases, four for $D_t$ and four for $D'_t$. We need not care about the value of $t$, 
		since all $D_t$, respectively $D'_t$, placements are mere shifts of each other.
		The defence-responsible guard is shown in black, while the rest is depicted in gray. 
		Also, notice that the pattern guards' cousins occupy positions $(w \pm 5\alpha, z \pm 5\beta)$ for $\alpha, \beta \in \mathbb{N}$, where $(w,z)$ is the new position of a pattern square guard.
		Then, $f(w, z) = f(w \pm 5\alpha, z \pm 5\beta)$ and $f'(w, z) = f'(w \pm 5\beta, z \pm 5\alpha)$, since the modulo $5$ operation cancels out the addition (subtraction) of $5\alpha$ and $5\beta$.
		A similar observation holds for stationary guards:
		we identify a model guard, say on position $(a,b)$, and then the rest of such guards are given by $(a \pm 5\alpha, b \pm 5\beta)$.
		Again, the $f(\cdot)$, respectively $f'(\cdot)$, values of all these vertices remain equal.
		For this reason, we focus below only on the pattern guards and one stationary guard and demonstrate that they share the same value of $f(\cdot)$, respectively $f'(\cdot)$. 
		
		We hereby consider a potential attack around a vertex $(x, y) \in D_t$.
		\begin{paragraph}{Attack on $(x-1, y)$ (i.e., on $SQ_3(x,y)$).}
			We apply Rotate-Square around $SQ_0(x, y)$.
			The four guards around $SQ_0(x, y)$ and the model stationary guard move as follows (Figure~\ref{fig:DtSQ3}):
			
			Let $P$ stand for the set of new positions given in Table~\ref{tab:table1}.
			The guards now occupy cells $(w,z) \in P$ where $f'(w, z) = 2x + y - 2 \;(\bmod\; 5) = 2x + y +3 \;(\bmod\; 5) = t'$.
			By this fact, we get $P \subseteq D'_{t'}$.
			Now, assume there exists a cell $(w,z) \notin P$, but $(w, z) \in D'_{t'}$.
			Without loss of generality, we assume $w \in [x-3, x+1]$ and $z \in [y-1, y+3]$, 
			since the configuration of the guards in this window is copied all over the grid by the symmetry of $D_t$ or $D'_t$ placements.
			Since $(w,z) \notin P$, this is a cell with no guard on it.
			However, by construction, any such cell is dominated by an adjacent vertex $(w_1, z_1)$ with $f'(w_1, z_1) = t'$.
			Then, by assumption, $f'(w,z) = f'(w_1, z_1) = t'$, which is a contradiction because, by definition of $f'(\cdot)$,
			two adjacent cells never have equal values. 
		\end{paragraph}
		
		All other cases can be proved in a similar fashion; for the details of each case, see Tables~\ref{tab:table2}--\ref{tab:table8}.
		Notice that an attack against a $D_t$ placement leads to a $D'_{t'}$ placement for some $t'$ and vice versa.
	\end{proof}
	
	\begin{theorem}\label{thm:INF}
		The guards eternally dominate $G_\infty$ by following the Rotate-Square strategy starting from an initial dominating set in $\mathcal{D}(G_\infty) \cup \mathcal{D}'(G_\infty)$.
	\end{theorem}
	\begin{proof}
		We prove by induction that the guards defend against any number of attacks and always maintain a placement in $\mathcal{D}(G_\infty) \cup \mathcal{D}'(G_\infty)$ after their turn.
		
		In the first step, the guards apply Rotate-Square and by Lemma~\ref{lem:step} they successfully defend against the first attack and now form another dominating set in $\mathcal{D}(G_\infty) \cup \mathcal{D}'(G_\infty)$.
		
		Assume that $i$ attacks have occurred and the guards have successfully defended against all of them by following Rotate-Square.
		That is, they occupy a configuration in $\mathcal{D}(G_\infty) \cup \mathcal{D}'(G_\infty)$.
		The $(i+1)$-st attack now occurs and the guards again follow Rotate-Square and therefore defend against the attack and form another dominating set in $\mathcal{D}(G_\infty) \cup \mathcal{D}'(G_\infty)$ (by Lemma~\ref{lem:step}).
	\end{proof}
	
	
	
	\begin{minipage}{\linewidth}
		
		
		\begin{table}[H]
			\centering
			\scalebox{0.9}{
				\begin{tabular}{| c |c | c | c|}
					\hline
					Guard & Old Position $(w, z)$ & New Position $(w',z')$ & $f'(w', z')$ $(\bmod\; 5)$ \\
					\hline
					defence-responsible &$(x,y)$ & $(x-1, y)$ & $2x + y - 2$\\
					pattern &$(x-2, y+1)$ & $(x-2, y+2)$ & $2x + y - 2$\\
					pattern &$(x-1,y +3)$ & $(x, y+3)$ & $2x + y + 3$\\
					pattern &$(x+1, y+2)$ & $(x+1, y+1)$ & $2x + y + 3$\\
					\hline
					\hline
					stationary &$(x-3, y-1)$ & $(x-3, y-1)$ & $2x + y - 2$\\ 
					\hline
				\end{tabular}}
				\caption{Attack on $(x-1, y)$ (rotate around $SQ_0(x, y)$); Figure~\ref{fig:DtSQ3}}
				\label{tab:table1} 
			\end{table}
			
			
			\begin{figure}[H]
				\centering
				\begin{minipage}{0.45\linewidth}
					\begin{subfigure}{\textwidth}
						\centering
						\includegraphics[scale = 0.65]{attack-sq3.pdf}
						\caption{Attack on bottom-right cell of $SQ_3(x, y)$} 
						\label{fig:DtSQ3-1}
					\end{subfigure}
				\end{minipage}
				\begin{minipage}{0.45\linewidth}
					\begin{subfigure}{\textwidth}
						\centering
						\includegraphics[scale = 0.65]{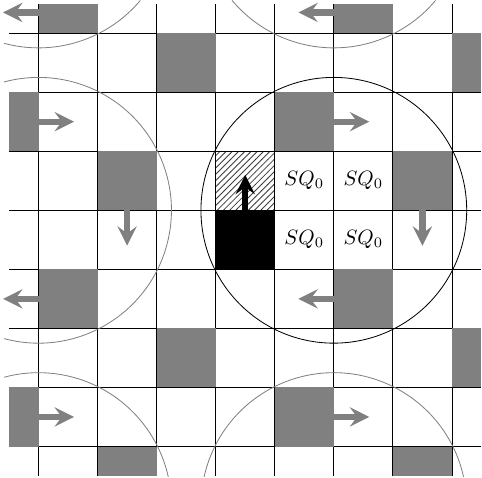}       
						\caption{Rotate-Square with $SQ_0(x, y)$ pattern square}
						\label{fig:DtSQ3-2}
					\end{subfigure}
				\end{minipage}
				\caption{Defending against an attack on $SQ_3(x, y)$}
				\label{fig:DtSQ3}
			\end{figure}
			
			
			\vspace{-2.5mm} %
			
			\begin{table}[H]
				\centering
				\scalebox{0.9}{
					\begin{tabular}{| c |c | c | c|}
						\hline
						Guard & Old Position $(w, z)$ & New Position $(w',z')$ & $f'(w', z')$ $(\bmod\; 5)$ \\
						\hline
						defence-responsible &$(x,y)$ & $(x, y-1)$ & $2x + y - 1$\\
						pattern &$(x-1, y-2)$ & $(x-2, y-2)$ & $2x + y - 1$\\
						pattern &$(x-3,y-1)$ & $(x-3, y)$ & $2x + y - 1$\\
						pattern &$(x-2, y+1)$ & $(x-1, y+1)$ & $2x + y - 1$\\
						\hline
						\hline
						stationary &$(x+1, y+2)$ & $(x+1, y+2)$ & $2x + y + 4$\\ 
						\hline
					\end{tabular}}
					\caption{Attack on $(x, y-1)$ (rotate around $SQ_3(x, y)$); Figure~\ref{fig:DtSQ2}}
					\label{tab:table2}
				\end{table}
				
				
				\begin{figure}[H]
					\centering
					\begin{minipage}{0.45\linewidth}
						\begin{subfigure}{\textwidth}
							\centering
							\includegraphics[scale = 0.65]{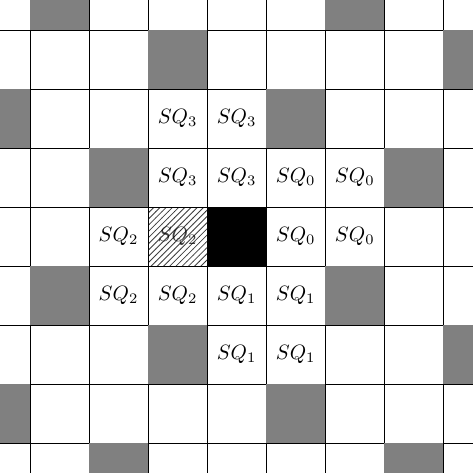}
							\caption{Attack on top-right cell of $SQ_2(x, y)$} 
							\label{fig:DtSQ2-1}
						\end{subfigure}
					\end{minipage}
					\begin{minipage}{0.45\linewidth}
						\begin{subfigure}{\textwidth}
							\centering
							\includegraphics[scale = 0.65]{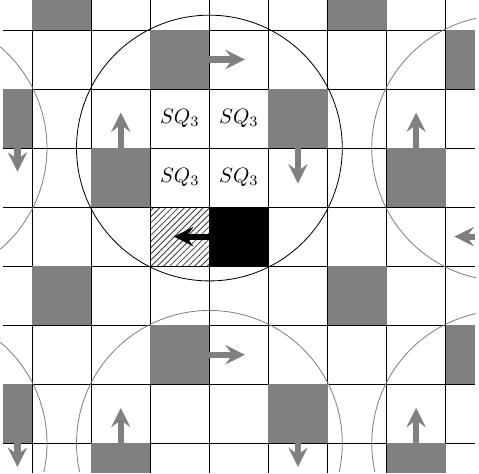}       
							\caption{Rotate-Square with $SQ_3(x, y)$ pattern square}
							\label{fig:DtSQ2-2}
						\end{subfigure}
					\end{minipage}
					\caption{Defending against an attack on $SQ_2(x, y)$}
					\label{fig:DtSQ2}
				\end{figure}
				
			\end{minipage}
			
			\begin{minipage}{\linewidth}
				
				
				\begin{table}[H]
					\centering
					\scalebox{0.9}{
						\begin{tabular}{| c |c | c | c|}
							\hline
							Guard & Old Position $(w, z)$ & New Position $(w',z')$ & $f'(w', z')$ $(\bmod\; 5)$ \\
							\hline
							defence-responsible & $(x,y)$ & $(x+1, y)$ & $2x + y + 2$\\
							pattern &$(x+2, y-1)$ & $(x+2, y-2)$ & $2x + y + 2$\\
							pattern &$(x+1,y-3)$ & $(x, y-3)$ & $2x + y - 3$\\
							pattern &$(x-1, y-2)$ & $(x-1, y-1)$ & $2x + y - 3$\\
							\hline
							\hline
							stationary &$(x-2, y+1)$ & $(x-2, y+1)$ & $2x + y - 3$\\ 
							\hline
						\end{tabular}}
						\caption{Attack on $(x+1, y)$ (rotate around $SQ_2(x, y)$); Figure~\ref{fig:DtSQ1}}
					\end{table}
					
					
					\begin{figure}[H]
						\centering
						\begin{minipage}{0.45\linewidth}
							\begin{subfigure}{\textwidth}
								\centering
								\includegraphics[scale = 0.65]{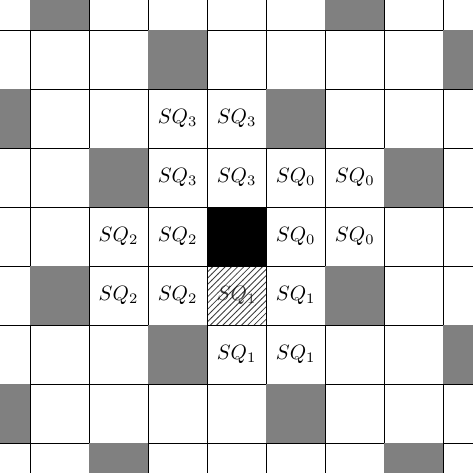} 
								\caption{Attack on top-left cell of $SQ_1(x, y)$} 
								\label{fig:DtSQ1-1}
							\end{subfigure}
						\end{minipage}
						\begin{minipage}{0.45\linewidth}
							\begin{subfigure}{\textwidth}
								\centering
								\includegraphics[scale = 0.65]{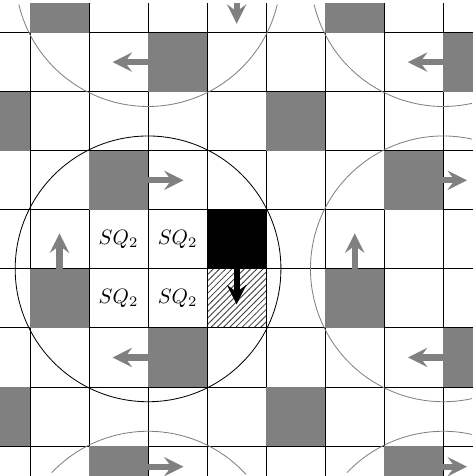}  
								\caption{Rotate-Square with $SQ_2(x, y)$ pattern square}
								\label{fig:DtSQ1-2}
							\end{subfigure}
						\end{minipage}
						\caption{Defending against an attack on $SQ_1(x, y)$}
						\label{fig:DtSQ1}
					\end{figure}
					
					
					\vspace{-2.5mm} %
					
					\begin{table}[H]
						\centering
						\scalebox{0.9}{
							\begin{tabular}{| c |c | c | c|}
								\hline
								Guard & Old Position $(w, z)$ & New Position $(w',z')$ & $f'(w', z')$ $(\bmod\; 5)$ \\
								\hline
								defence-responsible &$(x,y)$ & $(x, y+1)$ & $2x + y + 1$\\
								pattern &$(x+1, y+2)$ & $(x+2, y+2)$ & $2x + y + 1$\\
								pattern &$(x+3,y+1)$ & $(x+3, y)$ & $2x + y + 1$\\
								pattern &$(x+2, y-1)$ & $(x+1, y-1)$ & $2x + y + 1$\\
								\hline
								\hline
								stationary &$(x-1, y+3)$ & $(x-1, y+3)$ & $2x + y + 1$\\ 
								\hline
							\end{tabular}}
							\caption{Attack on $(x, y+1)$ (rotate around $SQ_1(x, y)$); Figure~\ref{fig:DtSQ0}}
						\end{table}

						
						\begin{figure}[H]
							\centering
							\begin{minipage}{0.45\linewidth}
								\begin{subfigure}{\textwidth}
									\centering
									\includegraphics[scale = 0.65]{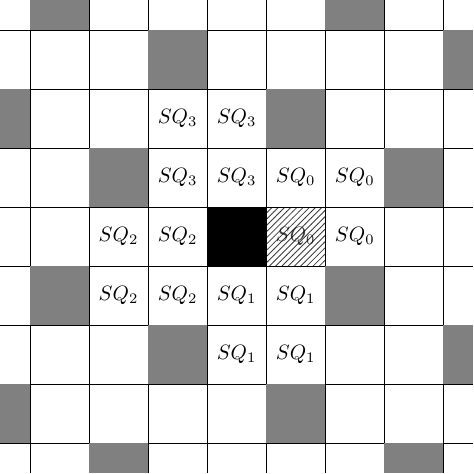} 
									\caption{Attack on bottom-left cell of $SQ_0(x, y)$} 
									\label{fig:DtSQ0-1}
								\end{subfigure}
							\end{minipage}
							\begin{minipage}{0.45\linewidth}
								\begin{subfigure}{\textwidth}
									\centering
									\includegraphics[scale = 0.65]{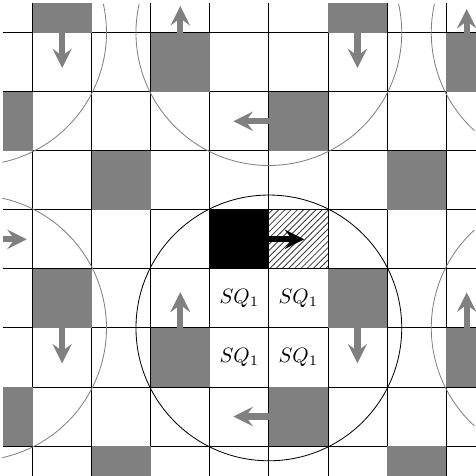}       
									\caption{Rotate-Square with $SQ_1(x, y)$ pattern square}
									\label{fig:DtSQ0-2}
								\end{subfigure}
							\end{minipage}
							\caption{Defending against an attack on $SQ_0(x, y)$}
							\label{fig:DtSQ0}
						\end{figure}
						
					\end{minipage}
					
					\begin{minipage}{\linewidth}

						
						
						\begin{table}[H]
							\centering
							\scalebox{0.9}{
								\begin{tabular}{| c |c | c | c|}
									\hline
									Guard & Old Position $(w, z)$ & New Position $(w',z')$ & $f(w', z')$ $(\bmod\; 5)$ \\
									\hline
									defence-responsible &$(x,y)$ & $(x-1, y)$ & $x + 2y - 1$\\
									pattern &$(x-2, y-1)$ & $(x-2, y-2)$ & $x + 2y - 1$\\
									pattern &$(x-1,y -3)$ & $(x, y-3)$ & $x + 2y -1$\\
									pattern &$(x+1, y-2)$ & $(x+1, y-1)$ & $x + 2y -1 $\\
									\hline
									\hline
									stationary &$(x-3, y+1)$ & $(x-3, y+1)$ & $x + 2y - 1$\\ 
									\hline
								\end{tabular}}
								\caption{Attack on $(x-1, y)$ (rotate around $SQ'_2(x, y)$); Figure~\ref{fig:DtSQ'3}}
							\end{table}
							
							
							\begin{figure}[H]
								\centering
								\begin{minipage}{0.45\linewidth}
									\begin{subfigure}{\textwidth}
										\centering
										\includegraphics[scale = 0.65]{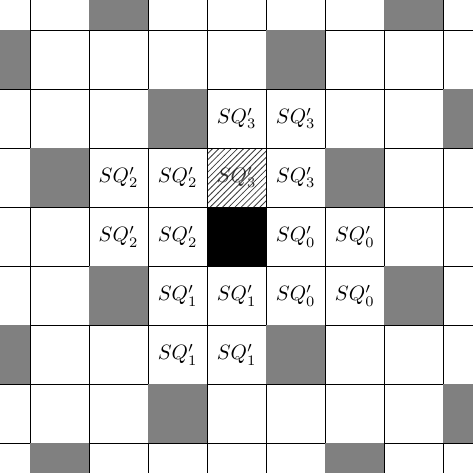}
										\caption{Attack on bottom-left cell of $SQ'_3(x, y)$} 
										\label{fig:DtSQ'3-1}
									\end{subfigure}
								\end{minipage}
								\begin{minipage}{0.45\linewidth}
									\begin{subfigure}{\textwidth}
										\centering
										\includegraphics[scale = 0.65]{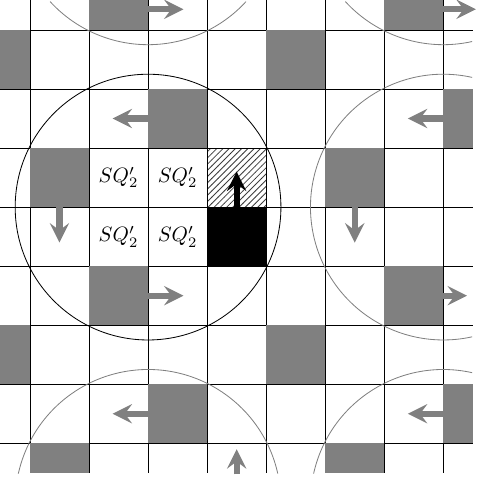}       
										\caption{Rotate-Square with $SQ'_2(x, y)$ pattern square}
										\label{fig:DtSQ'3-2}
									\end{subfigure}
								\end{minipage}
								\caption{Defending against an attack on $SQ'_3(x, y)$}
								\label{fig:DtSQ'3}
							\end{figure}
							
							
							\vspace{-2.5mm} %
							
							\begin{table}[H]
								\centering
								\scalebox{0.9}{
									\begin{tabular}{| c |c | c | c|}
										\hline
										Guard & Old Position $(w, z)$ & New Position $(w',z')$ & $f(w', z')$ $(\bmod\; 5)$ \\
										\hline
										defence-responsible &$(x,y)$ & $(x, y-1)$ & $x + 2y - 2$\\
										pattern &$(x+2, y+1)$ & $(x+1, y+1)$ & $x + 2y + 3$\\
										pattern &$(x+3,y-1)$ & $(x+3, y)$ & $x + 2y + 3$\\
										pattern &$(x+1, y-2)$ & $(x+2, y-2)$ & $x + 2y - 2$\\
										\hline
										\hline
										stationary &$(x-1, y+2)$ & $(x-1, y+2)$ & $x + 2y + 3$\\ 
										\hline
									\end{tabular}}
									\caption{Attack on $(x, y-1)$ (rotate around $SQ'_1(x, y)$); Figure~\ref{fig:DtSQ'2}}
								\end{table}
								
								
								\begin{figure}[H]
									\centering
									\begin{minipage}{0.45\linewidth}
										\begin{subfigure}{\textwidth}
											\centering
											\includegraphics[scale = 0.65]{attack-sq_2.pdf}
											\caption{Attack on bottom-right cell of $SQ'_2(x, y)$} 
											\label{fig:DtSQ'2-1}
										\end{subfigure}
									\end{minipage}
									\begin{minipage}{0.45\linewidth}
										\begin{subfigure}{\textwidth}
											\centering
											\includegraphics[scale = 0.65]{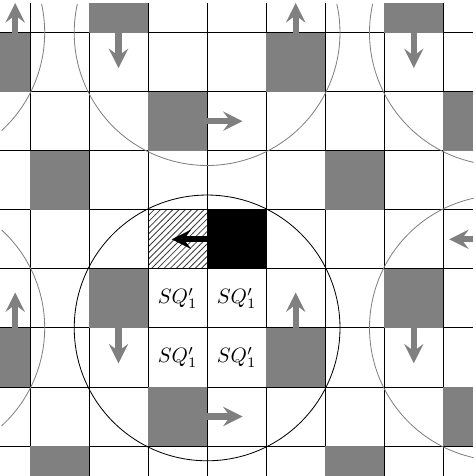}       
											\caption{Rotate-Square with $SQ'_1(x, y)$ pattern square}
											\label{fig:DtSQ'2-2}
										\end{subfigure}
									\end{minipage}
									\caption{Defending against an attack on $SQ'_2(x, y)$}
									\label{fig:DtSQ'2}
								\end{figure}
								
							\end{minipage}
							
							\begin{minipage}{\linewidth}
								
								
								\begin{table}[H]
									\centering
									\scalebox{0.9}{
										\begin{tabular}{| c |c | c | c|}
											\hline
											Guard & Old Position $(w, z)$ & New Position $(w',z')$ & $f(w', z')$ $(\bmod\; 5)$ \\
											\hline
											defence-responsible &$(x,y)$ & $(x+1, y)$ & $x + 2y + 1$\\
											pattern &$(x+2, y+1)$ & $(x+2, y+2)$ & $x + 2y + 1$\\
											pattern &$(x+1,y+3)$ & $(x, y+3)$ & $x + y + 1$\\
											pattern &$(x-1, y+2)$ & $(x-1, y+1)$ & $x + 2y + 1$\\
											\hline
											\hline
											stationary &$(x-2, y-1)$ & $(x-2, y-1)$ & $x + 2y - 4$\\ 
											\hline
										\end{tabular}}
										\caption{Attack on $(x+1, y)$ (rotate around $SQ'_0(x, y)$); Figure~\ref{fig:DtSQ'1}}
									\end{table}
									
									
									\begin{figure}[H]
										\centering
										\begin{minipage}{0.45\linewidth}
											\begin{subfigure}{\textwidth}
												\centering
												\includegraphics[scale = 0.65]{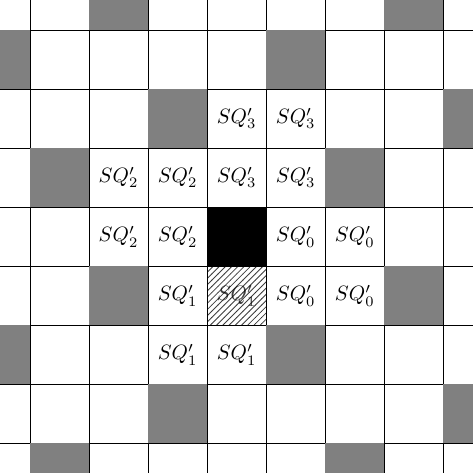}
												\caption{Attack on top-right cell of $SQ'_1(x, y)$} 
												\label{fig:DtSQ'1-1}
											\end{subfigure}
										\end{minipage}
										\begin{minipage}{0.45\linewidth}
											\begin{subfigure}{\textwidth}
												\centering
												\includegraphics[scale = 0.65]{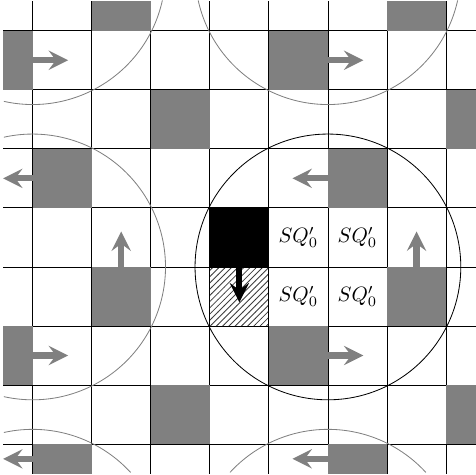}       
												\caption{Rotate-Square with $SQ'_0(x, y)$ pattern square}
												\label{fig:DtSQ'1-2}
											\end{subfigure}
										\end{minipage}
										\caption{Defending against an attack on $SQ'_1(x, y)$}
										\label{fig:DtSQ'1}
									\end{figure}

									
									\vspace{-2.5mm} %
									
									\begin{table}[H]
										\centering
										\scalebox{0.9}{
											\begin{tabular}{| c |c | c | c|}
												\hline
												Guard & Old Position $(w, z)$ & New Position $(w',z')$ & $f(w', z')$ $(\bmod\; 5)$ \\
												\hline
												defence-responsible &$(x,y)$ & $(x, y+1)$ & $x + 2y + 2$\\
												pattern &$(x-1, y+2)$ & $(x-2, y+2)$ & $x + 2y + 2$\\
												pattern &$(x-3,y+1)$ & $(x-3, y)$ & $x + 2y - 3$\\
												pattern &$(x-2, y-1)$ & $(x-1, y-1)$ & $x + 2y + -3$\\
												\hline
												\hline
												stationary &$(x+1, y+3)$ & $(x+1, y+3)$ & $x + 2y + 2$\\ 
												\hline
											\end{tabular}}
											\caption{Attack on $(x, y+1)$ (rotate around $SQ'_3(x, y)$); Figure~\ref{fig:DtSQ'0}}
											\label{tab:table8}
										\end{table}
										
										
										\begin{figure}[H]
											\centering
											\begin{minipage}{0.45\linewidth}
												\begin{subfigure}{\textwidth}
													\centering
													\includegraphics[scale = 0.65]{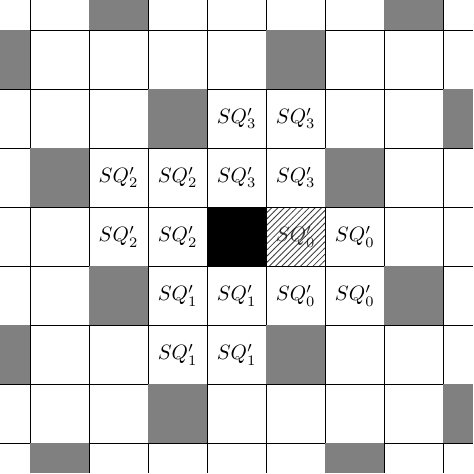}
													\caption{Attack on top-left cell of $SQ'_0(x, y)$} 
													\label{fig:DtSQ'0-1}
												\end{subfigure}
											\end{minipage}
											\begin{minipage}{0.45\linewidth}
												\begin{subfigure}{\textwidth}
													\centering
													\includegraphics[scale = 0.65]{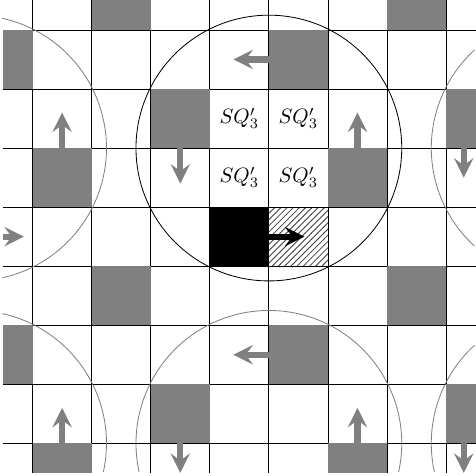}       
													\caption{Rotate-Square with $SQ'_3(x, y)$ pattern square}
													\label{fig:DtSQ'0-2}
												\end{subfigure}
											\end{minipage}
											\caption{Defending against an attack on $SQ'_0(x, y)$}
											\label{fig:DtSQ'0}
										\end{figure}
										
									\end{minipage}
									
									\section{Eternally Dominating Finite Grids}\label{sec:fin}
									We now apply the Rotate-Square strategy to finite grids, i.e., graphs of the form $P_{m}\Box P_{n}$.
									The initial idea is to follow the rules of the strategy, but to never leave any boundary or corner vertex unoccupied.
									A finite $m \times n$ grid consists of vertices $(i, j)$ where $i \in \{0, 1, 2, \ldots, m-1\}$ and $j \in \{0, 1, 2, \ldots, n-1\}$.
									Vertices $(0,x), (m-1, x), (y, 0), (y, n-1)$ for $x \in \{1, 2, \ldots, n-2\}$ and $y \in \{1, 2, \ldots, m-2\}$ are called \emph{boundary vertices}, while vertices $(0,0), (0, n-1), (m-1, 0), (m-1, n-1)$ are called \emph{corner vertices}.   
									Adjacencies are similar to the infinite grid case. 
									However, boundary vertices only have three neighbors, while corner vertices only have two.
									Let us consider $V(t) = D_t \cap (P_{m}\Box P_n)$ and $V'(t) = D'_t \cap (P_{m}\Box P_n)$, respectively. 
									We cite the following counting lemma from \cite{Chang}.
									
									\begin{lemma}[Lemma 2.2 \cite{Chang}]\label{lem:Vt}
										For all $t$, it holds $ \lfloor\frac{mn}{5}\rfloor \le |V(t)| \le \lceil\frac{mn}{5}\rceil $, and there exist $t_0, t_1$, such that
										$|V(t_0)| = \lfloor\frac{mn}{5}\rfloor \text{ and } |V(t_1)| = \lceil\frac{mn}{5}\rceil$.
									\end{lemma}
									
									The main observation in the proof of the above lemma is that there exist either $\lfloor \frac{m}{5} \rfloor$ or $\lfloor \frac{m}{5} \rfloor + 1$ $D_t$-vertices in one column of a $P_{m}\Box P_n$ grid.
									Then, a case-analysis counting provides the above bounds.
									The same observation holds for $D'_t$, since $f'(\cdot)$ is defined based on the same function $f: \mathbb{Z}^2 \rightarrow \mathbb{Z}_5$. 
									Thence, we can extend the above lemma for $D'_t$ cases with a similar proof.
									
									\begin{lemma}\label{lem:V1t}
										For all $t$, it holds $ \lfloor\frac{mn}{5}\rfloor \le |V'(t)| \le \lceil\frac{mn}{5}\rceil $, and there exist $t_0, t_1$, such that
										$|V'(t_0)| = \lfloor\frac{mn}{5}\rfloor \text{ and } |V'(t_1)| = \lceil\frac{mn}{5}\rceil$.
									\end{lemma}
									
									In order to study the domination of $P_m \Box P_n$, the analysis is based on examining $V(t)$, but for an extended $P_{m+2} \Box P_{n+2}$ mesh.
									Indeed, Chang \cite{Chang} showed the following.
									
									\begin{lemma}[Theorem 2.2 \cite{Chang}]\label{lem:Chang}
										For any $m, n \ge 8$,  $\gamma_{m,n} \le \lfloor \frac{(m+2)(n+2)}{5} \rfloor - 4$.
									\end{lemma}
									
									The result follows by picking an appropriate $D_t$ placement
									and forcing into the boundary of $P_m \Box P_n$ the guards on the boundary of $P_{m+2} \Box P_{n+2}$. 
									Moreover, Chang showed how to eliminate another four guards; one near each corner.
									
									Below, to facilitate the readability of our analysis, we focus on a specific subcase of finite grids.
									We demonstrate an m-eternal dominating strategy for $m \times n$ finite grids where $m \bmod 5 = n \bmod 5 = 2$ and then we improve upon it.
									Later, we extend to the general case.
									
									\subsection{A First Upper Bound: Full Boundary Cover}
									
									Initially, we place our guards on vertices belonging to $V(t) = D_t \cap (P_{m}\Box P_n)$ for some value $t \in \mathbb{Z}_5$.
									Unlike the approach in \cite{Chang}, we do not force inside any guards lying outside the boundary of $P_m \Box P_n$.
									Since a sequence of attacks may force the guards to any $V(t)$ or $V'(t)$ placement, i.e., for any value of $t$,
									we pick an initial guard placement, say $V(t_1)$, for which $|V(t_1)| = \lceil \frac{mn}{5} \rceil$ holds, to
									make sure that there are enough guards to maintain domination while transitioning from one placement to the next. 
									By Lemma~\ref{lem:Vt}, such a placement exists.
									Moreover, we cover the whole boundary by placing a guard on each unoccupied boundary or corner vertex.
									For an example, see Figure~\ref{fig:init}: black cells stand for guards which are members of a $D_t$ placement, whereas shaded cells denote the places where the extra guards are put.
									We refer to each of these added guards as a \emph{boundary guard}.
									This concludes the initial placement of the guards.
									
									The guards now follow Rotate-Square limited within the grid boundaries.
									For grid regions lying far from the boundary, Rotate-Square is applied in the same way as in the infinite grid case.
									For guard moves happening near the boundary or the corners, Rotate-Square's new placement demands can be satisfied by performing \emph{shifts of boundary guards}.
									In other words, a guard may need to step \emph{out of} the boundary, because he is (a cousin of) a pattern guard.
									Then, another guard steps \emph{into} the boundary to replace him, while the boundary guards between the into and out-of cells shift one step on the boundary.
									An example can be found in Figure~\ref{fig:shift-boundary-0}, where we partially view the area near the bottom-left corner of a finite grid. 
									The circles enclose repetitions of the pattern square movement.
									We focus our attention in the following two cases.
									\begin{itemize}
										\item Guard transitions within the boundary: 
										As indicated in the leftmost column in Figure~\ref{fig:shift-boundary-0}, to follow the pattern move, the two black guards have to move downward.
										However, since their movement does not force them outside the boundary and the boundary is fully occupied, there is no need to move in this case.
										We demonstrate this by removing the arrows in Figure~\ref{fig:shift-boundary-1}.
										
										\item Guard transitions into and out of the boundary:
										As indicated in the enclosed rectangle at the bottom-center in Figure~\ref{fig:shift-boundary-1}, following the pattern means a guard has to leave the boundary, whereas another has to enter it.
										To perform the pattern move, while maintaining a full boundary, we perform a shift of boundary guards as depicted in Figure~\ref{fig:shift}.
										Essentially, the three boundary guards between the two pattern guards shift one step to the left to both cover the unoccupied cell left by the pattern guard leaving the boundary and free a cell for the new position of the pattern guard entering the boundary.		
									\end{itemize}
									
									Overall, we refer to this slightly modified version of Rotate-Square as \emph{Finite Rotate-Square}.
									
									\begin{figure}
										\centering
										\includegraphics[scale = 0.55]{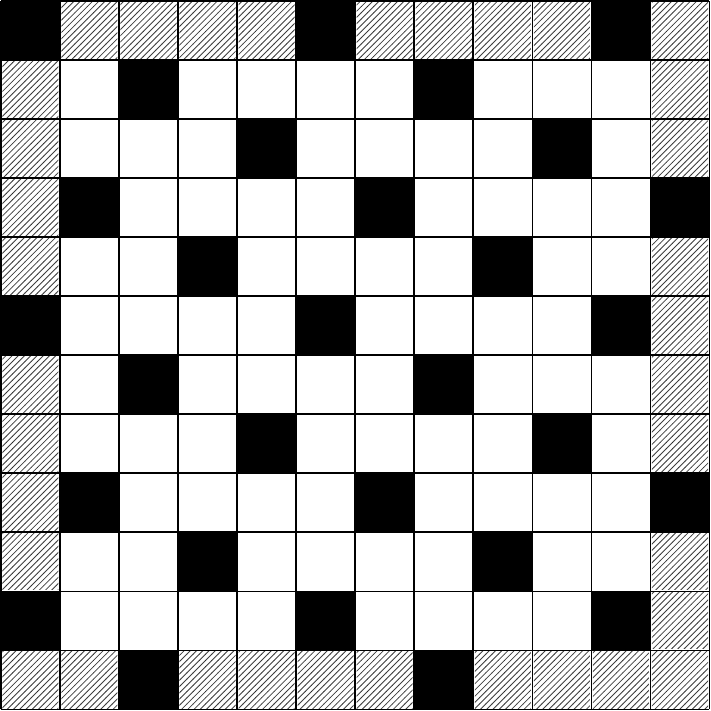} 
										\caption{An example of an initial $D_t$ placement for the guards in a $12\times12$ grid: black cells stand for $D_t$ guards, whereas shaded cells stand for extra boundary guards}
										\label{fig:init}
									\end{figure}
									
									\begin{figure}
										\centering
										\begin{minipage}{0.45\linewidth}
											\begin{subfigure}{\textwidth}
												\centering
												\includegraphics[scale = 0.65]{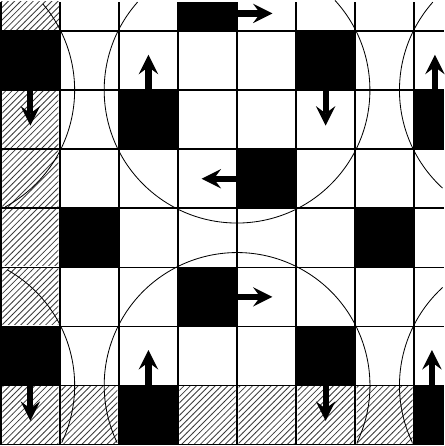}
												\caption{Circles enclosing pattern square repetitions near a corner of a finite grid} 
												\label{fig:shift-boundary-0}
											\end{subfigure}
										\end{minipage}
										\hfill
										\begin{minipage}{0.45\linewidth}
											\begin{subfigure}{\textwidth}
												\centering
												\includegraphics[scale = 0.65]{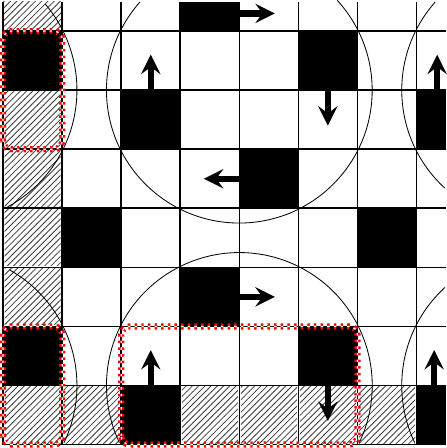}       
												\caption{Guard transitions within and into/ out of the boundary (within dashed rectangles)}
												\label{fig:shift-boundary-1}
											\end{subfigure}
										\end{minipage}
										\caption{An example of Finite Rotate-Square near the bottom-left corner of a finite grid}
										\label{fig:shift-boundary}
									\end{figure}
									
									\begin{figure}
										\centering
										\begin{minipage}{0.3\linewidth}
											\begin{subfigure}{\textwidth}
												\centering
												\includegraphics[scale = 0.6]{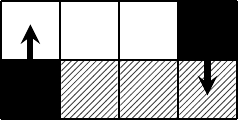}
												\caption{} 
												\label{fig:shift-1}
											\end{subfigure}
										\end{minipage}
										\begin{minipage}{0.3\linewidth}
											\begin{subfigure}{\textwidth}
												\centering
												\includegraphics[scale = 0.6]{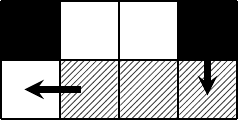}       
												\caption{}
												\label{fig:shift-2}
											\end{subfigure}
										\end{minipage}
										\begin{minipage}{0.3\linewidth}
											\begin{subfigure}{\textwidth}
												\centering
												\includegraphics[scale = 0.6]{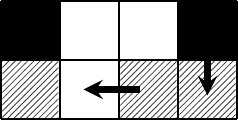}       
												\caption{}
												\label{fig:shift-3}
											\end{subfigure}
										\end{minipage}
										\\[1.5em]
										\begin{minipage}{0.3\linewidth}
											\begin{subfigure}{\textwidth}
												\centering
												\includegraphics[scale = 0.6]{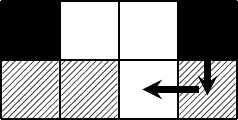}
												\caption{} 
												\label{fig:shift-4}
											\end{subfigure}
										\end{minipage}
										\begin{minipage}{0.3\linewidth}
											\begin{subfigure}{\textwidth}
												\centering
												\includegraphics[scale = 0.65]{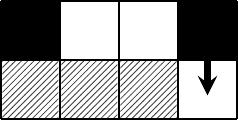}       
												\caption{}
												\label{fig:shift-5}
											\end{subfigure}
										\end{minipage}
										\begin{minipage}{0.3\linewidth}
											\begin{subfigure}{\textwidth}
												\centering
												\includegraphics[scale = 0.65]{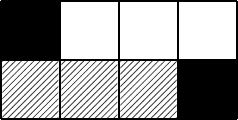}       
												\caption{}
												\label{fig:shift-6}
											\end{subfigure}
										\end{minipage}
										\caption{Demonstrating an example of boundary guards' shifting for the case of Figure~\ref{fig:shift-boundary-1}}
										\label{fig:shift}
									\end{figure}
									
									\begin{figure}
										\centering
										\begin{minipage}{0.45\textwidth}
											\begin{subfigure}{\textwidth}
												\hspace{-42pt}~\includegraphics[scale = 0.6]{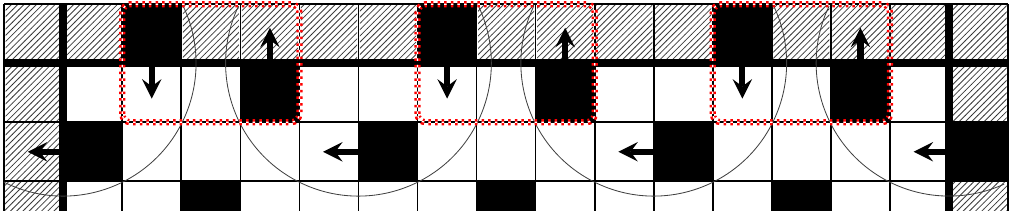}
												\caption{$\frac{n-2}{5} = 3$ pairs: guards of adjacent squares} 
												\label{fig:boundary-cut-0}
											\end{subfigure}
										\end{minipage}
										\\[0.5em]
										\begin{minipage}{0.45\textwidth}
											\begin{subfigure}{\textwidth}
												\hspace{-42pt}~\includegraphics[scale = 0.6]{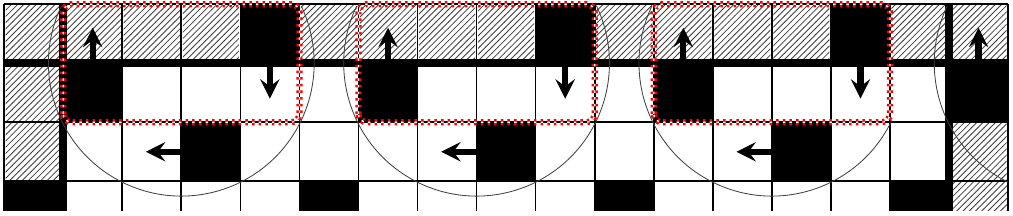}       
												\caption{$\frac{n-2}{5} = 3$ pairs: guards of same square}
												\label{fig:boundary-cut-1}
											\end{subfigure}
										\end{minipage}
										\caption{Examples of forming boundary-shifting pairs on the upper boundary of a grid ($n = 17$)}
										\label{fig:boundary-cut}
									\end{figure}
									
									\begin{lemma}\label{lem:finite-boundary}
										Assume $m \bmod 5 = n \bmod 5 = 2$ and that the guards follow Finite Rotate-Square, for an m-Eternal Domination game in $P_m \Box P_n$.
										Then, after every turn, their new placement $P$ is dominating, all boundary and corner vertices have a guard on them
										and, for some $t \in \mathbb{Z}_5$, there exists a set $V(t)$ (or $V'(t))$ such that $V(t) \subseteq P$ (or $V'(t) \subseteq P$).
									\end{lemma}
									
									\begin{proof}
										Consider the $(m-2) \times (n-2)$ subgrid that remains when we remove the boundary rows and columns.
										Since $m \bmod 5 = n \bmod 5 = 2$, $(m-2)$ and $(n-2)$ perfectly divide $5$.
										The latter means that each row, respectively column, of the subgrid has exactly $\frac{n-2}{5}$, respectively $\frac{m-2}{5}$, guards on it.
										Without loss of generality, consider row one neighboring the upper boundary row, which is row zero.
										Let us assume that a pattern square propagation forces a row-one guard to move into the boundary.
										Then, by symmetry of the pattern guards' move, there exists another guard on the boundary row zero that needs to move downward to row one.
										Notice that the same holds for each of the $\frac{n-2}{5}$  guards lying on row one, since the pattern guards' move propagates in hops of distance five.
										Movements in and out of the boundary alternate due to the shape of the pattern square.
										Moreover, we do not need to care about where the pattern square repetition is "cut" by the left/right boundary since, due to $n-2$ perfectly dividing $5$, there are exactly $\frac{n-2}{5}$ full pattern squares occurring subject to shifting.
										Consequently, we can apply the shifting procedure demonstrated in Figure~\ref{fig:shift} to apply the moves and maintain a full boundary, while preserving the number of guards on row one.
										In Figure~\ref{fig:boundary-cut}, we demonstrate examples of the above remarks: $\frac{n-2}{5}$ guards move from row one into row zero, whereas $\frac{n-2}{5}$ move vice versa.
										By the discussion above, it is always possible to form pairs of leaving/entering boundary guards and apply the shifting procedure demonstrated in Figure~\ref{fig:shift} either leftward or rightward.
										
										The new placement $P$ is dominating, since the $(m-2) \times (n-2)$ subgrid is dominated by any $V(t)$ or $V'(t)$ placement and 
										the boundary is always full of guards. 
										Moreover, since we follow a modified Rotate-Square, $P$ contains as a subset a vertex set $V(t)$ or $V'(t)$ after each guards' turn.
									\end{proof}

									\begin{lemma}\label{lem:bound}
										For $m, n \ge 7$ such that $m \bmod 5 = n \bmod 5 = 2$,
										$\gamma^\infty_{m,n} \le \frac{mn}{5} + \frac{8}{5}(m+n) - \frac{16}{5}$ holds.
									\end{lemma}
									
									\begin{proof}
										By inductively applying Lemma~\ref{lem:finite-boundary}, Finite Rotate-Square eternally dominates $P_m \Box P_n$.
										
										From the initial $V(t)$ placement, we get exactly $\frac{(m-2)(n-2)}{5}$ guards within $P_{m-2}\Box P_{n-2}$, since $(m-2)$ and $(n-2)$ perfectly divide $5$.
										Then, we need another $2(m+n) - 4$ guards to cover the whole boundary.
										Overall, the guards sum to $ \frac{(m-2)(n-2)}{5} + 2(m+n)- 4 = \frac{mn}{5} + \frac{8}{5}(m+n) - \frac{16}{5}$.
									\end{proof}
									
									
									\subsection{An Improved Upper Bound: Partial Boundary Cover}
									
									In the version of Finite Rotate-Square just presented above, the entirety of the boundary always remains covered. More specifically, five guards are placed for every sequence of five non-corner boundary vertices. Optimistically, we would like to lower the number of guards to two guards per every five boundary vertices. Then, compared to the standard domination number, this would provide only a constant-factor additive term.
									In this subsection, we prove an improved upper bound for Finite Rotate-Square by using three guards for each sequence of five non-corner boundary vertices. Furthermore, we discuss why having two guards would instead most likely fail for Finite Rotate-Square (or simple variations of it).
									
									\begin{lemma}\label{lem:imp-bound}
										For $m, n \ge 7$ such that $m \bmod 5 = n \bmod 5 = 2$,
										$\gamma^\infty_{m,n} \le \frac{mn}{5} + \frac{4}{5}(m+n)$ holds.
									\end{lemma}
									
									\begin{proof}
										First, let us take advantage of the condition $m \bmod 5 = n \bmod 5 = 2$ in order to reduce this family of grids to the $7\times7$ grid case.
										Imagine a general $m\times n$ grid where $m \bmod 5 = n \bmod 5 = 2$ holds.
										The non-boundary vertices can be partitioned into $\frac{(m-2)(n-2)}{5}$ subgrids of size $5\times5$, e.g., see Figure~\ref{fig:reduction}.
										Moreover, we add four guards, one in each corner, which never move throughout the execution of the strategy, since they can never leave the boundary.
										Then, we can partition each boundary row/column of the grid into sequences of length five.
										
										Now, notice that the far-from-the-boundary guards implement Rotate-Square and, due to the modulo $5$ use in the emergent $D_t$ and $D_t'$ placements, all $5\times5$ subgrids are copies of each other. 
										Moreover, for the same reason, all segments of length five on the same boundary row/column look identical at all times.
										Thence, we can contract all $5\times 5$ subgrids and side segments until a $7\times 7$ grid is left. Below, we provide a strategy for this special case.
										For $m, n > 7$, the strategy can be extracted by copying the $7\times 7$ strategy in each subgrid and boundary row/column segment.
										
										\begin{figure}
											\centering
											\includegraphics[scale = 0.55]{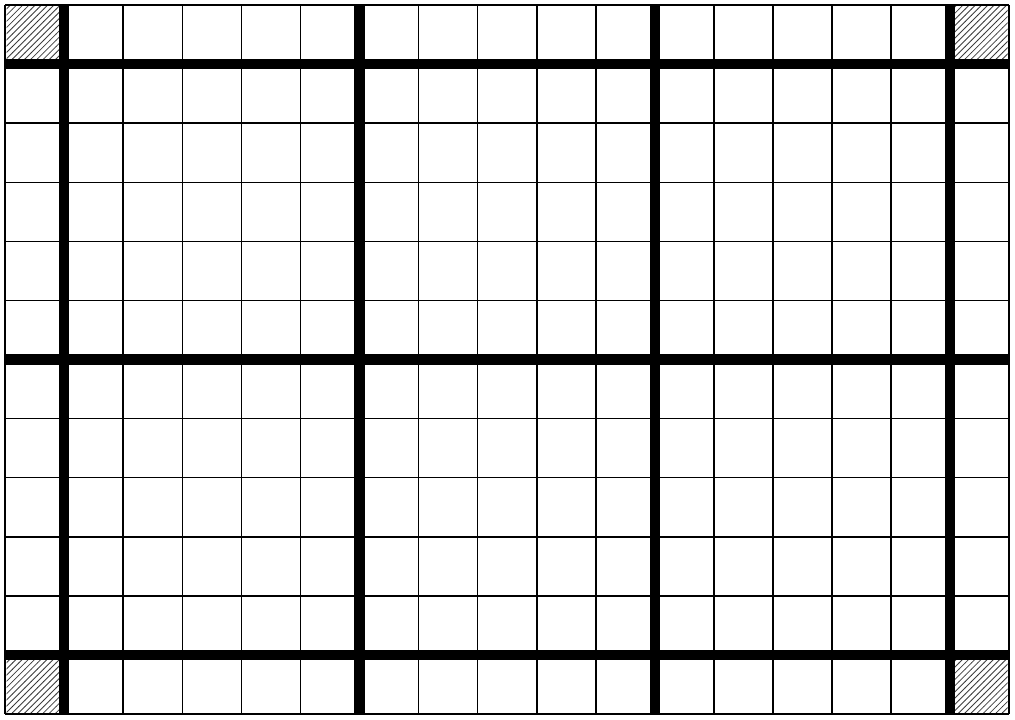}
											\caption{An example of partitioning a $17\times 12$ grid into $5\times5$ subgrids}
											\label{fig:reduction}
										\end{figure}
										
										Hence, to prove the bound, it suffices to provide an m-eternal domination strategy for the $7\times7$ grid with each corner always occupied by an immovable guard and three guards guarding each boundary row/column of length five.  
										In Figure~\ref{fig:imp}, we demonstrate such placements for the $7\times7$ grid: in column $A$ on the left, we depict all possible $D_t$ placements with the corresponding boundary cover, whereas, in column $B$ on the right, we depict all possible $D_t'$ placements. Moreover, in column $A$, we provide all the guard transitions for an attack to an unoccupied vertex. Transitions in column $B$ are omitted since all  guard movements are reversible. 
										
										By inductively applying the strategy of Figure~\ref{fig:imp}, this improved version of Finite Rotate-Square eternally dominates $P_m \Box P_n$, since the guards always form an $A$ or $B$ placement.
										From the initial $V(t)$ placement, we get exactly $ \frac{(m-2)(n-2)}{5}$ guards within $P_{m-2}\Box P_{n-2}$, since $(m-2)$ and $(n-2)$ perfectly divide $5$.
										Then, we require four guards for the corners and another $\frac{3}{5}(m-2)$, respectively $\frac{3}{5}(n-2)$, to cover a side of the grid.
										Overall, we get $ \frac{(m-2)(n-2)}{5} + 2\cdot\frac{3}{5}(m-2 + n-2) + 4 = \frac{mn}{5} + \frac{4}{5}(m+n)$ guards suffice for m-eternal domination.
									\end{proof} 
									
									\begin{figure}
										\centering
										\vspace{-55pt}
										\includegraphics[scale = 0.515]{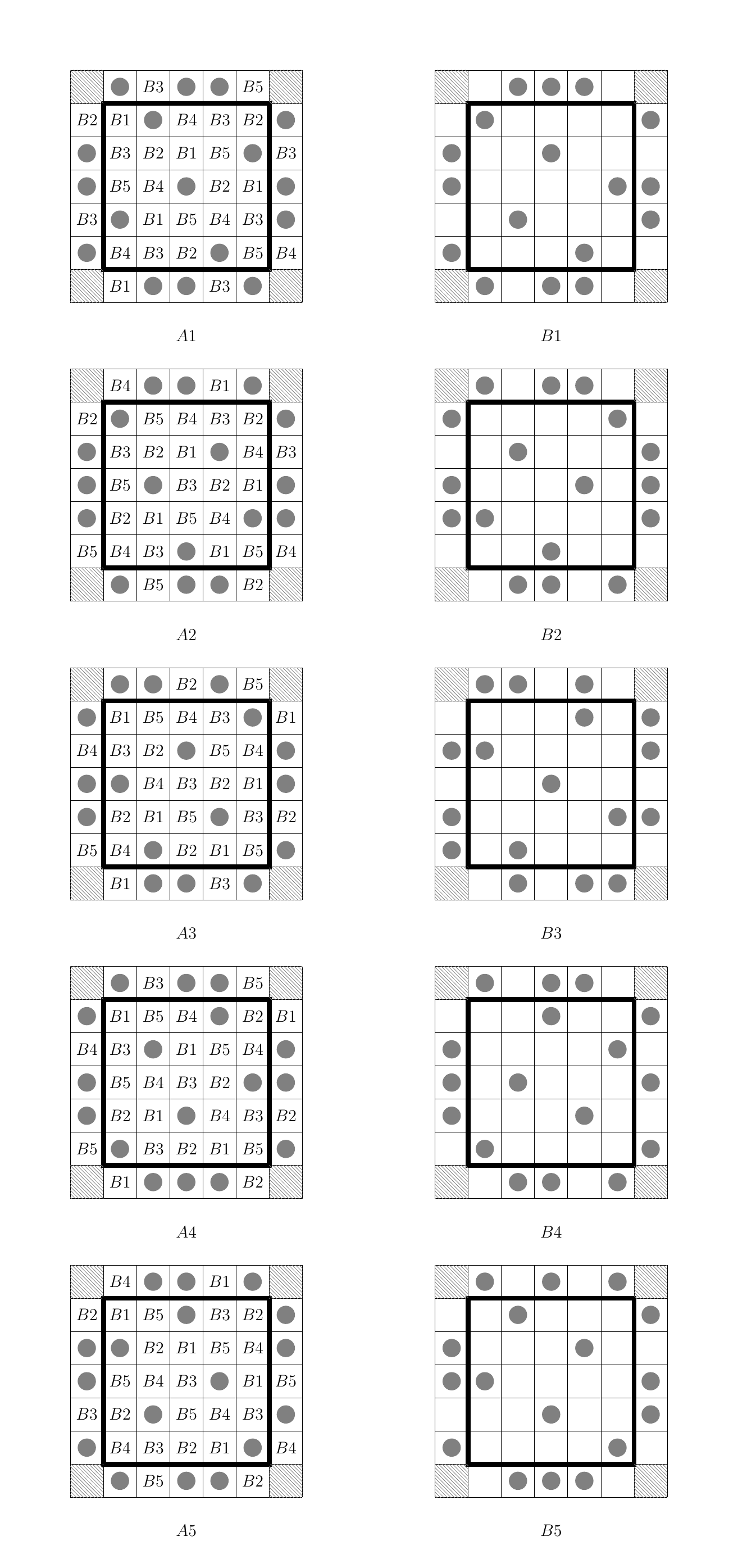}
										\caption{Improved Finite Rotate-Square (reduced to the $7\times7$ grid)}
										\label{fig:imp}
									\end{figure}
									
									
									The above proof is crucially based on the fact that $\gamma^\infty_{\mathrm{m}}(P_5) = 3$.
									Indeed, it is easy to verify that two guards cannot m-eternally dominate a path of length five.
									Therefore, a uniform approach as the one taken in the proof of Lemma~\ref{lem:imp-bound} is bound to fail. Furthermore, for non-uniform boundary guarding approaches, the problem seems to persist.
									In such approaches, boundary guards are not dedicated to a single $P_5$ segment of the side.
									Intuitively, the latter can easily lead to the creation of bias, meaning that eventually the extra corner guard (or any constant number of extra corner guards) has to move in order to assist with the protection of the boundary and leave the corner unoccupied.
									
									\subsection{A General Upper Bound}
									
									So far, we have focused on the special case $m \bmod 5 = n \bmod 5 = 2$ for which we provided an upper bound for the m-eternal domination number.
									We generalize this bound for arbitrary $m$, $n$ values.
									
									\begin{lemma}\label{lem:general}
										For $m, n \ge 7$, $\gamma^\infty_{m,n} \le \frac{mn}{5} + \mathcal{O}(m + n)$ holds.
									\end{lemma}
									\begin{proof}
										The idea behind this general bound is to "thicken" the boundary in the cases when $m \bmod 5 = n \bmod 5 = 2$ does not hold and then apply Finite Rotate-Square as above.
										More formally, one can identify an $(m-i)\times(n-j)$ subgrid, in the interior of the $m \times n$ grid, where $i, j \le 5$ such that $(m-i) \bmod 5 = (n-j) \bmod 5 = 2$ and execute the strategy there.
										For the rest of the rows and columns, they can be eternally secured by populating them with $\mathcal{O}(m + n)$ extra guards, e.g, in the worst-case, place one guard at each such cell.
									\end{proof}
									
									Gon\c{c}alves et al. \cite{Goncalves} showed $\gamma_{m,n} \ge \lfloor \frac{(m+2)(n+2)}{5} \rfloor - 4 $ for any $m, n \ge 16$.
									By combining this with Lemma~\ref{lem:Chang}, we get the exact domination number $\gamma_{m,n} =  \lfloor \frac{(m+2)(n+2)}{5} \rfloor - 4$ for $m, n \ge 16$.
									Then, by using Lemma~\ref{lem:general}, our main result follows. 
									
									\begin{theorem}\label{thm:main}
										For any $m, n \ge 16$, $\gamma^\infty_{m,n} \le \gamma_{m,n} + \mathcal{O}(m + n)$ holds. 
									\end{theorem}
									
									\section{Conclusions}\label{sec:con}
									
									We demonstrated a first strategy to m-eternally dominate general rectangular grids based on the repetition of a rotation pattern.
									Regarding further work, a more careful case analysis of the boundary may lead to improvements regarding the coefficient of the linear term.
									It is unclear whether this strategy can be used to obtain a constant additive gap between domination and m-eternal domination in large grids.
									Furthermore, the existence of a stronger lower bound than the trivial $\gamma^\infty_{m,n} \ge \gamma_{m,n}$ bound also remains open.
									

								\end{document}